\documentclass[reqno,10pt,letterpaper]{amsart}

\usepackage{lipsum}
\usepackage{amsmath}
\usepackage{amssymb}
\usepackage{amsthm}
\usepackage{mathrsfs}
\usepackage{accents}
\usepackage{calc}
\usepackage{arydshln}
\usepackage{upgreek}
\usepackage{slashed}
\usepackage{xifthen}
\usepackage{graphicx}
\usepackage{subcaption}
\usepackage{longtable}
\usepackage[inline]{enumitem}

\usepackage{xr}
\usepackage{tikz}

\usepackage{xcolor}
\definecolor{winered}{rgb}{0.6,0,0}
\definecolor{lessblue}{rgb}{0,0,0.7}

\usepackage[pdftex,colorlinks=true,linkcolor=winered,citecolor=lessblue,urlcolor=lessblue,breaklinks=true,bookmarksopen=true]{hyperref}

\usepackage{listings}
\lstdefinestyle{mathematica}{
  language=Mathematica,
  basicstyle=\ttfamily\footnotesize,
  columns=fullflexible,
  keepspaces=true,
  showstringspaces=false,
  breaklines=true,
  breakatwhitespace=false,
  frame=single,
  framerule=0.3pt,
  framesep=6pt,
  rulecolor=\color{black!20},
  backgroundcolor=\color{black!2},
  commentstyle=\itshape\color{black!65},
  keywordstyle=\bfseries,
  numbers=none,
  tabsize=2
}

\makeatletter
\newcommand{\myitem}[2]{\item[\rm(#2)]\def\@currentlabel{#2}\label{#1}}
\makeatother

\usepackage{titletoc}

\makeatletter

\def\@tocline#1#2#3#4#5#6#7{
\begingroup
  \par
    \parindent\z@ \leftskip#3 \relax \advance\leftskip\@tempdima\relax
                  \rightskip\@pnumwidth plus 4em \parfillskip-\@pnumwidth
    \ifcase #1 % sections
       \vskip 0.6em \hskip 0em % add a little vspace before
       \or
       \or \hskip 0em % subsections
       \or \hskip 1em % subsubsections
    \fi%
    #6
    \nobreak\relax{\leavevmode\leaders\hbox{\,.}\hfill}
    \hbox to\@pnumwidth {\@tocpagenum{#7}}
  \par
\endgroup
}

 \def\l@section{\@tocline{0}{0pt}{0pc}{}{}}

\renewcommand{\tocsection}[3]{%
  \indentlabel{\@ifnotempty{#2}{ % for numbered sections
    \ignorespaces\bfseries{#2. #3}}}
  \indentlabel{\@ifempty{#2}{\ignorespaces\bfseries{#3}}{}} % for unnumbered sections
    \vspace{1.5pt}}

\renewcommand{\tocsubsection}[3]{%
  \indentlabel{\@ifnotempty{#2}{
    \ignorespaces#2. #3}}
  \indentlabel{\@ifempty{#2}{\ignorespaces #3}{}}
    \vspace{1.5pt}}

\renewcommand{\tocsubsubsection}[3]{%
  \indentlabel{\@ifnotempty{#2}{
    \ignorespaces#2. #3}}
  \indentlabel{\@ifempty{#2}{\ignorespaces #3}{}}
    \vspace{1.5pt}}

\makeatother

\makeatletter
\def\@nomenstarted{0}
\newlength{\@nomenoldtabcolsep}

\newcommand{\nomenstart}
  {%
    \def\@nomenstarted{1}%
    \setlength{\@nomenoldtabcolsep}{\tabcolsep}%
    \setlength{\tabcolsep}{3.5pt}%
    \begin{longtable}{p{0.11\textwidth} p{0.86\textwidth}}%found by hand
  }

\newcommand{\nomenitem}[2]{%
    \ifcase\@nomenstarted%
      \or % if nomenstarted=1, do nothing
      \or \\ % if nomenstarted=2, add newline to previous one
    \fi%
    #1\,{\leavevmode\leaders\hbox{\,.}\hfill} & #2%
    \def\@nomenstarted{2}%
  }%
\newcommand{\nomenend}
  {\\%
      \end{longtable}%
      \setlength{\tabcolsep}{\@nomenoldtabcolsep}%
      \def\@nomenstarted{0}%
  }
\makeatother

\makeatletter
\newcommand{\BIG}{\bBigg@{3.5}}
\newcommand{\vast}{\bBigg@{4}}
\newcommand{\Vast}{\bBigg@{5}}
\newcommand{\VAST}[1]{\bBigg@{#1}}
\makeatother

\allowdisplaybreaks

\numberwithin{equation}{section}
\numberwithin{figure}{section}
\newtheorem{thm}{Theorem}[section]

\newtheorem{prop}[thm]{Proposition}
\newtheorem{lemma}[thm]{Lemma}
\newtheorem{cor}[thm]{Corollary}

\newtheorem*{thm*}{Theorem}
\newtheorem*{prop*}{Proposition}
\newtheorem*{cor*}{Corollary}
\newtheorem*{conj*}{Conjecture}

\theoremstyle{definition}
\newtheorem{definition}[thm]{Definition}

\theoremstyle{remark}
\newtheorem{rmk}[thm]{Remark}

\makeatletter
\newcommand{\fakephantomsection}{%
  \Hy@MakeCurrentHref{\@currenvir.\the\Hy@linkcounter}
  \Hy@raisedlink{\hyper@anchorstart{\@currentHref}\hyper@anchorend}%
  \Hy@GlobalStepCount\Hy@linkcounter%
}
\makeatother

\newcommand{\mc}{\mathcal}

\newcommand{\cC}{\mc C}
\newcommand{\cD}{\mc D}

\newcommand{\cF}{\mc F}

\newcommand{\cO}{\mc O}

\newcommand{\cR}{\mc R}
\newcommand{\cS}{\mc S}

\newcommand{\ms}{\mathscr}

\newcommand{\sD}{\ms D}

\newcommand{\C}{\mathbb{C}}
\newcommand{\N}{\mathbb{N}}
\newcommand{\R}{\mathbb{R}}
\newcommand{\Z}{\mathbb{Z}}

\newcommand{\Sph}{\mathbb{S}}

\newcommand{\fp}{\mathfrak{p}}

\newcommand{\slDelta}{\slashed{\Delta}{}}

\newcommand{\slnabla}{\slashed{\nabla}{}}

\renewcommand{\Re}{\operatorname{Re}}
\renewcommand{\Im}{\operatorname{Im}}

\newcommand{\eps}{\epsilon}

\newcommand{\la}{\langle}

\newcommand{\ol}{\overline}
\newcommand{\pa}{\partial}
\newcommand{\dd}{{\mathrm d}}
\newcommand{\ra}{\rangle}
\newcommand{\spec}{\operatorname{spec}}

\newcommand{\wh}{\widehat}

\newcommand{\xra}{\xrightarrow}

\newcommand{\ubar}[1]{\underaccent{\bar}#1}%{\mkern1.5mu\underline{\mkern-1.5mu #1\mkern-1.5mu}\mkern1.5mu}

\newcommand{\CI}{\cC^\infty}

\newcommand{\Ric}{\mathrm{Ric}}

\newcommand{\bhm}{M}

\newcommand{\openbigpmatrix}[1]
  {%
    \def\@bigpmatrixsize{#1}%
    \addtolength{\arraycolsep}{-#1}%
    \begin{pmatrix}%
  }
\newcommand{\closebigpmatrix}
  {%
    \end{pmatrix}%
    \addtolength{\arraycolsep}{\@bigpmatrixsize}%
  }

\newlength{\enummargin}

\newcommand{\usref}[1]{{\upshape\ref{#1}}}

\DeclareGraphicsExtensions{.mps}

\makeatletter
\newcommand*{\fwbw}[1]{\expandafter\@fwbw\csname c@#1\endcsname}
\newcommand*{\@fwbw}[1]{\ifcase #1 \or {\rm fw}\or {\rm bw}\fi}
\AddEnumerateCounter{\fwbw}{\@fwbw}
\makeatother

\begin{document}

%%%%%%%%%%%%%%%%%%%%%%%%%%%%%%%%%%%%%%%%%%%%%%%%%%%%%%%%%%%%%%%%%%%%%%
% title page
\title[Mode stability of Kerr--de~Sitter]{Mode stability for the Klein--Gordon equation on subextremal Kerr--de~Sitter spacetimes}

\date{\today. Original version: August 26, 2026.}

% 83C57: black holes
% 35L05: wave equation
% 35P25: scattering theory
% 35C20: asymptotic expansions

\author{Peter Hintz}
\address{Department of Mathematics, Pennsylvania State University, 54 McAllister St, State College,\newline PA 16801, United States}
\email{phintz@psu.edu}

\begin{abstract}
  We prove mode stability for the Klein--Gordon equation on Kerr--de~Sitter black hole spacetimes in the full subextremal range and for the range $m_{\rm KG}^2\in[0,\Lambda\sqrt{2}]$ of squared scalar field masses that includes, in particular, the minimally coupled case $m_{\rm KG}=0$ and the conformally coupled case $m_{\rm KG}^2=\frac{2}{3}\Lambda$ (i.e., the Teukolsky equation for spin $0$). In combination with recent work by Petersen--Vasy \cite{PetersenVasySubextremal}, this proves, unconditionally, that sufficiently regular solutions of the scalar wave equation decay exponentially fast to constants, and in fact to zero for nonzero $m_{\rm KG}$.
\end{abstract}

\maketitle

%%%%%%%%%%%%%%%%%%%%%%%%%%%%%%%%%%%%%%%%%%%%%%%%%%%%%%%%%%%%%%%%%%%%%%
\section{Introduction}
\label{SI}

For parameters $\Lambda>0$ (cosmological constant), $\bhm>0$ (black hole mass), and $a\in\R$ (specific angular momentum), the Kerr--de~Sitter (KdS) metric is given by
\[
  g = -\frac{\mu(r)}{b^2\varrho^2(r,\theta)}\bigl(\dd t-a\sin^2\theta\,\dd\phi\bigr)^2 + \varrho^2(r,\theta)\Bigl(\frac{\dd r^2}{\mu(r)} + \frac{\dd\theta^2}{c(\theta)}\Bigr) + \frac{c(\theta)\sin^2\theta}{b^2\varrho^2(r,\theta)}\bigl((r^2+a^2)\,\dd\phi-a\,\dd t\bigr)^2
\]
in Boyer--Lindquist coordinates, where, introducing $\lambda:=\frac{\Lambda}{3}$,
\begin{equation}
\label{EqIFns}
\begin{split}
  \mu(r) &:= (r^2+a^2)(1-\lambda r^2) - 2\bhm r, \\
  \varrho^2(r,\theta) &:= r^2+a^2\cos^2\theta, \\
  b &:= 1+\lambda a^2, \\
  c(\theta) &:= 1+\lambda a^2\cos^2\theta.
\end{split}
\end{equation}
It is a solution of the Einstein vacuum equation $\Ric(g)-\Lambda g=0$ \cite{CarterHamiltonJacobiEinstein}. The parameters $(\Lambda,\bhm,a)$ are called \emph{subextremal} if the quartic polynomial $\mu(r)$ has four distinct real roots
\begin{equation}
\label{EqIRoots}
  r_- < r_C < r_e < r_c.
\end{equation}
(Necessarily $r_-<0\leq r_C$; see Lemma~\ref{LemmaGRoots}.) See Figure~\ref{FigIParam}.

\begin{figure}[!ht]
\centering
\includegraphics{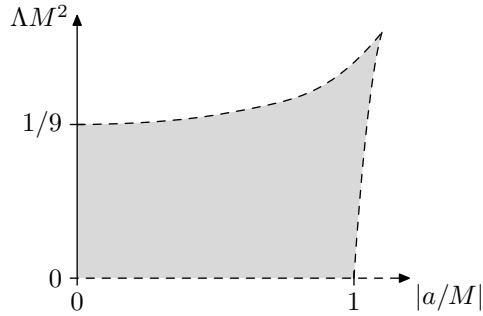}
\caption{The parameter space of subextremal KdS black holes. Mode stability had previously been shown in a neighborhood of the intervals $[0,1)$ and $[0,1/9)$ along the two axes.}
\label{FigIParam}
\end{figure}

The nonlinear stability of subextremal KdS spacetimes was proved by the author with Vasy \cite{HintzVasyKdSStability} in the regime of very small $a$ (see also \cite{FangKdS}), and in the full subextremal range with Petersen and Vasy \cite{HintzPetersenVasyKdS} \emph{conditional on the validity of mode stability for the ungauged Einstein equation} (see \cite{HintzPetersenVasyKdS} for the definition of this notion). In this paper, we take a first step towards verifying this mode stability unconditionally: rather than consider the (linearized) Einstein equation, we study here the scalar wave operator $\Box_g=-|g|^{-\frac12}\pa_\mu|g|^{\frac12}g^{\mu\nu}\pa_\nu$, explicitly given by
\begin{equation}
\label{EqIBox}
  \varrho^2\Box_g = -\pa_r\mu(r)\pa_r - \frac{1}{\sin\theta}\pa_\theta c(\theta)\sin\theta\,\pa_\theta + \frac{b^2}{\mu(r)}\bigl((r^2+a^2)\pa_t+a\pa_\phi\bigr)^2 - \frac{b^2}{c(\theta)\sin^2\theta}(\pa_\phi+a\sin^2\theta\,\pa_t)^2,
\end{equation}
on a subextremal KdS spacetime, and the more general Klein--Gordon operator
\[
  \Box_g + \lambda\nu
\]
for a range of scalar field masses $\lambda\nu$. Our main result is:

\begin{thm}[Mode stability, I: massless scalar waves]
\label{ThmI}
  Mode stability holds for $\Box_g$ for all subextremal KdS spacetimes. More precisely, there do not exist (purely oscillatory or exponentially growing) nontrivial \emph{outgoing mode solutions} $e^{-i\sigma t}u(r,\theta,\phi)$, $u\neq 0$, at frequencies $\sigma\in\C$ with $\Im\sigma\geq 0$, $\sigma\neq 0$; and the space of \emph{generalized outgoing mode solutions at frequency $0$} is spanned by the constant function $1$.
\end{thm}

\begin{thm}[Mode stability, II: massive scalar waves]
\label{ThmII}
  For all $0<\nu\leq 3\sqrt{2}$ (equivalently, $0<\frac{m_{\rm KG}^2}{\Lambda}\leq\sqrt 2$), mode stability holds for $\Box_g+\lambda\nu$ (equivalently, $\Box_g+m_{\rm KG}^2$) in the closed upper half plane, including at frequency $0$. In particular, mode stability holds for the Teukolsky equation with spin $\mathfrak{s}=0$ (which is the special case $\nu=2$).
\end{thm}

The upper bound on $\nu$ is not sharp (for example, mode stability holds for $\nu>0$ when $a=0$); but since it suffices to include the Teukolsky $\mathfrak{s}=0$ case, we do not try to optimize it. The analogue of Theorem~\ref{ThmII} for the Teukolsky equation with \emph{nonzero} integer spin $\mathfrak{s}$ (the cases $\mathfrak{s}=\pm 2$ being of primary relevance for the Einstein equation) remains open. Prior to the present work, mode stability for the wave equation on KdS spacetimes was rigorously known only near two boundaries of the parameter space:
\begin{itemize}
\item Dyatlov \cite{DyatlovQNM} proved mode stability in the slowly rotating regime and for arbitrary $\nu\geq 0$ (i.e., $|a/\bhm|\ll 1$, with the required smallness depending on the dimensionless quantities $\Lambda\bhm^2$ and $\nu$) via perturbation theory from the (simple) Schwarzschild--de~Sitter case $a=0$. Vasy's approach \cite{VasyMicroKerrdS} to the meromorphic continuation of the inverse of the spectral family $\wh{\Box_g}(\sigma)+\lambda\nu$ of $\Box_g+\lambda\nu$ via robust Fredholm theory yielded another proof of this result.
\item The author \cite{HintzKdSMS} treated the case of very light black holes (i.e., $\Lambda\bhm^2\ll 1$, with the required smallness depending on $|a/\bhm|$ and $\nu$) via perturbation theory off the known mode stability for the (massless) scalar wave equation on asymptotically flat Kerr black holes, which for real $\sigma$ is a celebrated result by Whiting \cite{WhitingKerrModeStability} as well as Shlapentokh-Rothman \cite{ShlapentokhRothmanModeStability} and Andersson--Ma--Paganini--Whiting \cite{AnderssonMaPaganiniWhitingModeStab}.\footnote{Mode stability in $\Im\sigma>0$, for unseparated modes, is then a consequence of the main theorem of \cite{DafermosRodnianskiShlapentokhRothmanDecay} (which proves decay for arbitrary solutions of the wave equation on subextremal Kerr). Another proof via a direct continuity argument on the spectral side is given in \cite[\S{3.9}]{HintzKdSMS}.} See Casals--Teixeira da Costa \cite{CasalsTeixeiradCModes} (building on Hatsuda \cite{HatsudaTeukolskyAlt} and Aminov--Grassi--Hatsuda \cite{AminovGrassiHatsudaQNM}) and Hollands--Ishibashi--Zahn \cite{HollandsIshibashiZahnKdS} for interpretations of the crucial \emph{Whiting transformation}, and the very recent work by Petersen--Vasy \cite{PetersenVasyModeStab} for another perspective on this in the spirit of positive commutator arguments.
\end{itemize}
For general subextremal KdS parameters, \cite{CasalsTeixeiradCModes} obtained a partial mode stability result for the Teukolsky equation that covered part, but not all, of the superradiant frequency range (where standard Wronskian arguments fail). As of yet, no generalization of Whiting's transformation has been found that would yield full mode stability in the KdS setting (see \cite{UmetsuKdS,SuzukiTakasugiUmetsuKdS}); mode stability for (most of) the superradiant frequency range (see~\S\ref{SN}, and also \cite{TachizawaMaedaKdSSuperradiance}) had remained elusive. Comprehensive numerical and semi-analytic investigations by Yoshida--Uchikata--Futamase \cite{YoshidaUchikataFutamaseKdS} and Novaes--Marinho--Lencs\'es--Casals \cite{NovaesMarinhoLencsesKdSQNM} provided evidence for mode stability.

Our result does not cover extremal KdS black holes (i.e., the right and top boundaries in Figure~\ref{FigIParam}). In the extremal Kerr case, mode stability was shown outside of two exceptional frequencies ($0$ and $\frac{m}{2\bhm}$ where $m\in\Z$ is the azimuthal frequency) by Teixeira da Costa \cite{TeixeiradCModes}. Mode stability for the Teukolsky equation on subextremal Kerr--\emph{anti} de~Sitter black holes was shown by Graf--Holzegel \cite{GrafHolzegelKerrAdSModes}.

That mode stability for rotating black holes is a delicate issue is demonstrated by results by Shlapentokh-Rothman \cite{ShlapentokhRothmanBlackHoleBombs} and Moschidis \cite{MoschidisSuperradiant} proving the existence of QNMs in the upper half plane for the Klein--Gordon equation on Kerr with suitable scalar field mass parameters or seemingly mild potential/metric modifications. This strongly suggests, in particular, that some restriction on $\nu$ is necessary in Theorem~\ref{ThmII}.

\bigskip

Before recalling the notion of an outgoing mode solution, we first note that the above expression of $g$ becomes singular at the roots of $\mu$. Fix the (real analytic) function $J=J(r)=2\frac{r-r_e}{r_c-r_e}-1$, the key property being
\[
  J(r_e)=-1,\quad J(r_c)=+1,
\]
and define coordinates $t_*,\phi_*$ by
\begin{equation}
\label{EqICoords}
  t_* := t-T(r),\quad \phi_* := \phi - \Phi(r),\qquad
  T'(r)=(r^2+a^2)\frac{b J(r)}{\mu(r)},\quad \Phi'(r)=a\frac{b J(r)}{\mu(r)}.
\end{equation}
In the coordinates $t_*,r,\theta,\phi_*$, the metric $g$ is then easily seen to extend analytically across $r=r_e$ and $r=r_c$. The hypersurfaces $\{r=r_e\}$ and $\{r=r_c\}$ in the extended spacetime manifold $\R_{t_*}\times(r_C,\infty)\times\Sph^2_{\theta,\phi_*}$ are then the \emph{event} and \emph{cosmological horizon}, respectively.

\begin{definition}[Outgoing mode solution]
\label{DefIMode}
  Fix $\delta>0$ such that $r_e-\delta>r_C$, and set
  \[
    X := [r_e-\delta,r_c+\delta]\times\Sph^2.
  \]
  An \emph{outgoing mode solution} of $\Box_g+\lambda\nu$ with frequency $\sigma\in\C$ is a smooth function $u\in\CI(X)$ such that
  \[
    (\Box_g+\lambda\nu)\bigl( e^{-i\sigma t_*}u(r,\theta,\phi_*) \bigr) = 0.
  \]
  If there exists a nonzero such $u$, we say that $\sigma$ is a \emph{quasinormal mode} (QNM). A \emph{generalized outgoing mode solution at frequency $0$} is a function $\sum_{j=0}^d t_*^j u_j(r,\theta,\phi_*)$ in the kernel of\footnote{This notion will only be relevant in this paper for $\nu=0$.} $\Box_g+\lambda\nu$ where $d\geq 0$ and $u_0,\ldots,u_d\in\CI(X)$.
\end{definition}

As shown by Petersen--Vasy \cite{PetersenVasySubextremal}, Theorems~\ref{ThmI} and \ref{ThmII} imply:

\begin{thm}[Expansion of waves]
\label{ThmIExp}
  On any subextremal KdS spacetime, there exist $\eps>0$ such that the following holds for all $s\geq 1$: the solution of $\Box_g u=f$ with $f\in e^{-\eps t_*}H^s(\R_{t_*}\times X)$, and with $u,f$ vanishing for $t_*\ll -1$, satisfies
  \[
    u(t_*,x) - c \in e^{-\eps t_*}H^s
  \]
  for some $c$ (depending linearly on $f$). The analogous statement holds for solutions of initial value problems with initial data (at a smooth spacelike hypersurface transversal to the future event and cosmological horizons) lying in $H^s\oplus H^{s-1}$. This remains true, with $c=0$ (i.e., exponential decay to $0$), for $\Box_g+\lambda\nu$ when $\nu$ lies in the range stated in Theorem~\usref{ThmII}.
\end{thm}

See also \cite[\S{3.3} (and Conjecture 4)]{ZworskiResonanceReview}. For small black holes, the optimal value of $\eps$ is computed in \cite{HintzKdSMS} to be $\sqrt{\Lambda/3}+o(1)$ as $\Lambda\bhm^2\to 0$ when $\nu=0$, with $|a/\bhm|<1$ bounded away from $1$; and in this case \cite{HintzKdSMS} obtains precise information even about QNMs with $\Im\sigma\gtrsim-\sqrt\Lambda$. The main result of \cite{HintzKdSMS} also covers the case $\nu\neq 0$; for $\nu>0$, it gives as the optimal value for $\eps$ the number $\Re(\frac32-\sqrt{\frac94-\nu})\sqrt{\Lambda/3}+o(1)$. Note the discontinuity as $\nu\to 0$; this is due to the QNM $0$ of the massless wave equation moving into the lower half plane upon increasing $\nu$. --- In the present paper, we obtain no information about QNMs with $\Im\sigma<0$.

\bigskip

The main work in the proof of Theorem~\ref{ThmI} is to treat the case of \emph{real} $\sigma$, as putative QNMs in the upper half plane cannot escape from there upon dialing the KdS parameters down to SdS parameters, for which mode stability in $\Im\sigma\geq 0$ is straightforward (\S\ref{SPf}). The case $\sigma=0$ is easily handled directly (\S\ref{S0}), so we focus now on $\sigma\in\R\setminus\{0\}$. After separation into azimuthal modes $e^{i m\phi_*}$ and spheroidal harmonics, the task is to prove the absence of nontrivial outgoing solutions $R=R(r)$ (see~\eqref{EqNBC}) of the radial ODE~\eqref{EqNODE} on $(r_e,r_c)$; and for this it suffices to prove the vanishing of the leading-order coefficients of $R$ at $r=r_\bullet$, where $R\sim B_\bullet|r-r_\bullet|^{i\kappa_\bullet}$ for some $\kappa_\bullet\in\R$. A standard Wronskian argument accomplishes this outside of the \emph{superradiant frequency range} $\sigma\in (\frac{a m}{r_c^2+a^2},\frac{a m}{r_e^2+a^2})$.

Unlike Whiting's argument \cite{WhitingKerrModeStability,ShlapentokhRothmanModeStability,AnderssonMaPaganiniWhitingModeStab}, which is based on a remarkable integral transform related to the radial Fourier transform, our new arguments for superradiant frequencies take place entirely in ``physical space,'' i.e., configuration space, and they are elementary in that they only involve integration by parts. We conjugate the radial ODE by a canonical oscillatory factor (that undoes the outgoing oscillations of $R$), with a sign in the phase that depends on whether one works near $r_e$ or $r_c$; the phases match up at an interior point $r=r_0\in(r_e,r_c)$. The radial ODEs on $[r_e,r_0]$ and $[r_0,r_c]$ are then of the form $(-\pa_r\mu\pa_r+L\mp 2 i\nabla_q)R_\pm=0$ where $\nabla_q$ is a formally skew-adjoint first order operator, and $L$ is a spheroidal eigenvalue. We use standard energy estimates with the multiplier $f_\pm\nabla_q$ on both sides of $r_0$ to prove a chain of estimates
\[
  0 \leq \Re(R'\ol{R})(r_0) \leq -C|B_c|^2
\]
from the monotonicity of two currents. Here $B_c$ is the boundary value of $R_+$ at $r_c$, and $C>0$ is a constant (see Corollary~\ref{CorFCc} and the arguments after Proposition~\ref{PropFCe}); so we get $B_c=0$ and thus $R=0$. The construction of $f_\pm$ satisfying the required monotonicity (see Lemma~\ref{LemmaFC}) is the technical heart of the paper (\S\ref{SF}). On the interval $[r_0,r_c]$, an explicit choice works (\S\ref{SsFCc}). On the interval $[r_e,r_0]$, the monotonicity is equivalent to a constrained ordinary differential inequality (see Lemma~\ref{LemmaFCeEquiv}), which we can solve either explicitly (\S\ref{SssCeI}) or via a piecewise argument (\S\ref{SssCeII})---first saturating the constraint, then solving the ordinary differential \emph{equality}---depending on the size of the angular eigenvalue. The positivity of various explicit polynomial expressions involving the KdS parameters in the full subextremal parameter range is the ``algebraic miracle'' that makes these constructions work.

In the Klein--Gordon case of Theorem~\ref{ThmII}, essentially the same strategy works, except the direct analogue of the $|R_\pm|^2$-term of the flux in the massless case is not necessarily pointwise negative on $[r_0,r_c]$, the failure becoming more severe with increasing $\nu$. When $\nu$ is not too large, we are, however, able to borrow some positivity of the $|R'_\pm|^2$-term (akin to a Hardy inequality) to complete the argument in an analogous fashion (\S\ref{SKG}).

%%%%%%%%%%%%%%%%%%%%%%%%%%%%%%%%%%%%%%%%%%%%%%%%%%
\subsection*{Declaration on AI use}

ChatGPT 5.6 Pro was used extensively in the exploratory phase of this project.

%%%%%%%%%%%%%%%%%%%%%%%%%%%%%%%%%%%%%%%%%%%%%%%%%%
\subsection*{Data availability statement}

Data sharing is not applicable to this article as no datasets were generated or analysed during the current study.

%%%%%%%%%%%%%%%%%%%%%%%%%%%%%%%%%%%%%%%%%%%%%%%%%%
\subsection*{Acknowledgments}

The author is grateful to the U.S.~National Science Foundation for support under grant DMS-2554160. I am grateful to Felipe Hernandez for an inspiring conversation about AI. I would like to dedicate this paper to my Georg.

%%%%%%%%%%%%%%%%%%%%%%%%%%%%%%
\section{The subextremal parameter space}
\label{SG}

Recall the notation~\eqref{EqIRoots} for the roots of $\mu$ in~\eqref{EqIFns}.

\begin{lemma}[Properties of the roots]
\label{LemmaGRoots}
  The roots of $\mu(r)$ for a subextremal Kerr--de~Sitter spacetime satisfy $r_-<0\leq r_C<r_e<r_c$ and $r_-=-(r_C+r_e+r_c)$.
\end{lemma}
\begin{proof}
  The second statement follows from the vanishing of the cubic term of $\mu(r)$.

  Consider the first statement. When $a=0$, then $\mu=r\bigl(r-\lambda r^3-2\bhm\bigr)$ has a root $r_C=0$, and the cubic factor $r-\lambda r^3-2\bhm$ is negative at $r=0$, is convex for $r<0$ (with second derivative $-6\lambda r>0$) and tends to $+\infty$ as $r\to-\infty$, thus has exactly one negative root $r_-$; and the remaining two real roots must then be positive. In the case $a\neq 0$, we have $\mu(r)\to-\infty$ as $r\to\pm\infty$ and $\mu(0)=a^2>0$, so the number of positive roots is odd. Suppose there were only one positive root of $\mu$; then there would need to be three distinct negative roots $r_0<r_1<r_2<0$ of $\mu$. Since $\mu(r_2)=0<\mu(0)$, we must have $\mu'(r_2)>0$, so since $\mu'(0)=-2\bhm<0$ we conclude that $\mu'$ has roots $r'_0\in(r_0,r_1)$, $r'_1\in(r_1,r_2)$, and $r'_2\in(r_2,0)$. Therefore $\mu''$ has two distinct negative roots, and thus $\mu'''=-24\lambda r$ would have a negative root; a contradiction.
\end{proof}

Conversely, given $0\leq r_C<r_e<r_c$, we can read off the parameters $\lambda=\frac{\Lambda}{3}$, $|a|$, and $\bhm$ from matching the coefficients of
\begin{equation}
\label{EqGParamExp}
\begin{split}
  &-\lambda(r-r_C)(r-r_e)(r-r_c)(r+r_C+r_e+r_c) \\
  &\quad = -\lambda r^4 + \lambda r^2(r_C^2+r_e^2+r_c^2+r_C r_e+r_C r_c+r_e r_c) \\
  &\quad \qquad - \lambda r(r_C+r_e)(r_C+r_c)(r_e+r_c) + \lambda r_C r_e r_c(r_C+r_e+r_c) \\
  &\quad =: -\lambda r^4 + \lambda w r^2 - \lambda s r + \lambda t
\end{split}
\end{equation}
with those of $(r^2+a^2)(1-\lambda r^2)-2\bhm r=-\lambda r^4+(1-\lambda a^2)r^2 - 2\bhm r+a^2$, as follows: note that $w,s>0$, and $t>0$ unless $r_C=0$ in which case $t=0$. In the case $t=0$ (so $r_C=0$), this gives $a^2=0$ and thus $\lambda w=1$, so $\lambda=w^{-1}$ and then $\bhm=\lambda s/2=s/(2 w)$. When $t>0$ (so $r_C>0$), we have $\lambda t=a^2>0$, so $\lambda>0$, and then $\lambda w=1-\lambda a^2=1-\lambda\cdot\lambda t$ has the unique positive solution $\lambda=-\frac{w}{2 t}+\sqrt{\frac{w^2}{4 t^2}+\frac{1}{t}}$, which then determines $\bhm=\lambda s/2$ and $|a|=\sqrt{\lambda t}$. In summary:

\begin{cor}[Parameterization]
\label{CorGParam}
  The above prescriptions provide a one-to-one correspondence between subextremal Kerr--de~Sitter parameters $(\Lambda,\bhm,a)$ with $a\geq 0$ and real numbers $r_-<0\leq r_C<r_e<r_c$ with $r_-=-(r_C+r_e+r_c)$.
\end{cor}

\begin{cor}[Some elementary bounds]
\label{CorGBound}
  We have $a^2\leq r_C r_e$, with strict inequality when $a\neq 0$, and $\mu(r)\leq (r-r_e)(r-r_C)$ for $r\in[r_e,r_c]$.
\end{cor}
\begin{proof}
  If $a=0$, then $r_C=0$, so equality $a^2=r_C r_e=0$ holds. If $a\neq 0$, then $r_C>0$, and since $a^2=\lambda t=\lambda r_C r_e r_c(r_C+r_e+r_c)$, the strict inequality is equivalent to $\lambda r_c(r_C+r_e+r_c)<1$. But $r_c(r_C+r_e+r_c)<w$; the first claim thus follows from $\lambda w=1-\lambda a^2\leq 1$, with strict inequality for $a\neq 0$.

  For the second claim, we write $\mu=(r-r_e)(r-r_C)\tilde\mu$, $\tilde\mu:=-\lambda(r-r_-)(r-r_c)$. We then compute
  \[
    \tilde\mu(r_e) = \lambda(r_e r_c-r_e^2-r_-r_c+r_-r_e) = \lambda(r_e r_c+r_C r_c+r_c^2 - 2 r_e^2 - r_C r_e) < \lambda w \leq 1
  \]
  and observe that $\tilde\mu'=-\lambda(2 r - r_- - r_c) = -\lambda(2 r+r_C+r_e) < 0$.
\end{proof}

We also recall an elementary scaling law (e.g., from \cite[\S{1.3}]{HintzKdSMS}): pullback along the scaling map $(t_*,r,\theta,\phi_*)\mapsto(s t_*,s r,\theta,\phi_*)$, $s>0$, maps the KdS metric with parameters $(\Lambda,\bhm,a)$, resp.\ a mode solution with frequency $\sigma$ to $s^2$ times the KdS metric with parameters $(\Lambda s^2,\bhm/s,a/s)$, resp.\ a mode solution with frequency $s\sigma$; under this scaling, the location of the event horizon changes from $r_e$ to $r_e/s$. Choosing $s=r_e$ thus shows that
\begin{equation}
\label{EqGNormalize}
  \text{\parbox{0.8\textwidth}{\it for the proof of Theorem~\usref{ThmI}, it suffices to consider the case $r_e=1$.}}
\end{equation}
We only use this reduction for verifying polynomial inequalities using computer algebra systems.

We record some polynomial inequalities that we will use later.

\begin{lemma}[Inequality \#1]
\label{LemmaGIneq1}
  For all $r_0\in[r_e,r_c]$ and $r\in(r_e,r_c)$, we have
  \[
    2\mu(r) - (r-r_0)\mu'(r) > 0.
  \]
\end{lemma}
\begin{proof}
  If $\mu'(r)=0$, then $\mu(r)>0$ for $r\in(r_e,r_c)$; so suppose $\mu'(r)\neq 0$. The function $f(r,r_0):=2\mu(r)-(r-r_0)\mu'(r)$ is affine in $r_0$. When $r$ is such that $\mu'(r)>0$, its minimum is attained at the left endpoint $r_0=r_e$. Now, in terms of the quantities
  \[
    x = \frac{r-r_e}{r_c-r_e}\in(0,1),\quad
    x_C = \frac{r_e-r_C}{r_c-r_e}>0,\quad
    x_- = \frac{r_e-r_-}{r_c-r_e}>0,
  \]
  we compute
  \[
    \frac{f(r,r_e)}{\mu(r)} = 2 - \Bigl(\frac{r-r_e}{r-r_-}+\frac{r-r_e}{r-r_C}+1+\frac{r-r_e}{r-r_c}\Bigr) = \frac{1}{1-x} - \frac{x}{x+x_-} - \frac{x}{x+x_C}.
  \]
  The common numerator is $2 x^3+x^2(x_-+x_C-1)+x_- x_C$. To show that this is positive, it suffices to show that $x_-+x_C>1$; but this is equivalent to $2 r_e-r_--r_C>r_c-r_e$, or $3 r_e>r_c+r_C+r_-=-r_e$.

  When $\mu'(r)<0$, the minimum of $f(r,r_0)$ over $r_0$ is attained at $r_0=r_c$; and we then compute
  \[
    \frac{f(r,r_c)}{\mu(r)} = 1 + \Bigl(\frac{r_c-r}{r-r_-}+\frac{r_c-r}{r-r_C}+\frac{r_c-r}{r-r_e}\Bigr) > 0.\qedhere
  \]
\end{proof}

%%%%%%%%%%%%%%%%%%%%%%%%%%%%%%%%%%%%%%%%%%%%%%%%%%%%%%%%%%%%%%%%%%%%%%
\section{Reduction to the radial ODE; boundary conditions}
\label{SN}

\emph{In this section, we only consider $\sigma\in\R$.} We first recall the standard reduction of the equation satisfied by mode solutions at real frequencies to a radial ODE.

Suppose $e^{-i\sigma t}u(r,\theta,\phi)$ is a nonzero mode solution of $\Box_g$. Then $u$ satisfies $\wh{\Box_g}(\sigma)u=0$ where $\wh{\Box_g}(\sigma)$ is obtained from $\Box_g$ by replacing $\pa_t$ by $-i\sigma$. Next, $\Box_g$ and thus $\wh{\Box_g}(\sigma)$ commute with pullbacks by rotations $\cR_\alpha\colon\phi\mapsto\phi+\alpha$. Given a smooth function $v(\theta,\phi)$ on $\Sph^2$ and a number $m\in\Z$,
\[
  v_m(\theta,\phi) := \frac{1}{2\pi}\int_0^{2\pi} e^{-i m\alpha}(\cR_\alpha^*v)(\theta,\phi)\,\dd\alpha
\]
is again smooth; this is the $e^{i m\phi}$-term of the Fourier series of $v$ in $\phi$. Applying this to $u(r,\theta,\phi)$ for each $r$ and observing that the $e^{i m\phi}$-component of $u$ still lies in $\ker\wh{\Box_g}(\sigma)$, it thus suffices to exclude mode solutions of the form $e^{-i\sigma t}e^{i m\phi}u(r,\theta)$. The equation satisfied by this (new) $u(r,\theta)$ is
\begin{align}
\label{EqNPDE}
  &-\pa_r\mu(r)\pa_r u(r,\theta) - \frac{b^2}{\mu(r)}\bigl((r^2+a^2)\sigma-a m\bigr)^2 u(r,\theta) + \slDelta_{\sigma,m}u(r,\theta) = 0, \\
\label{EqNAngOp}
  &\qquad \slDelta_{\sigma,m} := -\frac{1}{\sin\theta}\pa_\theta c(\theta)\sin\theta\,\pa_\theta + \frac{b^2}{c(\theta)\sin^2\theta}(m-a\sigma\sin^2\theta)^2.
\end{align}
Let us restore $\pa_\phi$ in the angular operator by introducing $\slDelta_\sigma:=-\frac{1}{\sin\theta}\pa_\theta c(\theta)\sin\theta\,\pa_\theta-\frac{b^2}{c(\theta)\sin^2\theta}(\pa_\phi-i a\sigma\sin^2\theta)^2$; then $\slDelta_\sigma$ is elliptic and symmetric (since $\sigma\in\R$) and commutes with $\pa_\phi$, and thus there exists a complete joint orthonormal basis of eigenfunctions (all of which are smooth). Keeping only those in $\ker(\pa_\phi-i m)$ gives a complete orthogonal basis of eigenfunctions of $\slDelta_{\sigma,m}$ in the space $L^2_m(\Sph^2):=L^2(\Sph^2)\cap\ker(\pa_\phi-i m)$. Since the PDE~\eqref{EqNPDE}, multiplied by $e^{i m\phi}$, commutes with projections onto $\slDelta_{\sigma,m}$-eigenspaces, we may thus assume that $e^{i m\phi}u(r,\theta)=R(r)S(\theta,\phi)$ where $S(\theta,\phi)=e^{i m\phi}\cS(\theta)\in L^2_m(\Sph^2)$ is an eigenfunction of $\slDelta_{\sigma,m}$ with eigenvalue $L\in\R$; therefore, the radial function $R(r)$ satisfies the ODE
\begin{equation}
\label{EqNODE}
  -(\mu R')' - \frac{q(r)^2}{\mu(r)}R + L R = 0,\quad q(r):=b\bigl((r^2+a^2)\sigma-a m\bigr),\quad L\in\spec\slDelta_{\sigma,m}.
\end{equation}

The original mode solution, which we may now assume to lie in $\ker(\pa_\phi-i m)\cap\ker(\slDelta_{\sigma,m}-L)$, takes the form $u(t,r,\theta,\phi)=e^{-i\sigma t}e^{i m\phi}\cS(\theta)R(r)$ where $\cS(\theta)\in\CI((0,\pi))$ (and $e^{i m\phi}\cS(\theta)\in\CI(\Sph^2)$). The outgoing condition demands that
\[
  u(t,r,\theta,\phi) = e^{-i\sigma t_*}e^{i m\phi_*}\cS(\theta)\tilde R(r),\quad \tilde R\in\CI([r_e,r_c]),
\]
so $R=e^{i(\sigma T(r)-m\Phi(r))}\tilde R\in e^{i(\sigma T(r)-m\Phi(r))}\CI([r_e,r_c])$. But from~\eqref{EqICoords}, we obtain
\begin{align*}
  T(r) &\equiv -\frac{b(r_e^2+a^2)}{\mu'(r_e)}\log(r-r_e) \bmod \CI([r_e,r_c)), \\
  \Phi(r) &\equiv -\frac{b a}{\mu'(r_e)}\log(r-r_e) \bmod \CI([r_e,r_c))
\end{align*}
and
\begin{align*}
  T(r) &\equiv -\frac{b(r_c^2+a^2)}{|\mu'(r_c)|}\log(r_c-r) \bmod \CI((r_e,r_c]), \\
  \Phi(r) &\equiv -\frac{b a}{|\mu'(r_c)|}\log(r_c-r) \bmod \CI((r_e,r_c]).
\end{align*}
(Note that $\mu'(r_e)$, $|\mu'(r_c)|=-\mu'(r_c)>0$.) Thus, the outgoing condition is equivalent to
\begin{equation}
\label{EqNBC}
  R(r) \in (r-r_e)^{-i q(r_e)/\mu'(r_e)}\CI([r_e,r_c)) \cap (r_c-r)^{-i q(r_c)/|\mu'(r_c)|}\CI((r_e,r_c]).
\end{equation}
In particular, the quantities
\begin{equation}
\label{EqNBC2}
  B_e := \lim_{r\searrow r_e} (r-r_e)^{i q(r_e)/\mu'(r_e)}R(r),\quad
  B_c := \lim_{r\nearrow r_c} (r_c-r)^{i q(r_c)/|\mu'(r_c)|}R(r)
\end{equation}
are well-defined. The main input for the proof of Theorem~\ref{ThmI} is:

\begin{thm}[No nontrivial radial solutions: real nonzero $\sigma$]
\label{ThmNNo}
  For $0\neq\sigma\in\R$, every solution $R=R(r)$ of~\eqref{EqNODE} that satisfies the outgoing condition~\eqref{EqNBC} must vanish.
\end{thm}

The previous reduction shows that this implies Theorem~\ref{ThmI} for $\sigma\in\R\setminus\{0\}$. (The remaining frequencies will be handled in~\S\S\ref{S0}--\ref{SPf}.) Since replacing $(\sigma,m)$ by $(-\sigma,-m)$ and $(a,m)$ by $(-a,-m)$ leaves the ODE~\eqref{EqNODE} unchanged, we may without loss assume that $m\geq 0$ and $a\geq 0$. We first recall the simple proof of Theorem~\ref{ThmNNo} in the non-superradiant frequency range:

\begin{proof}[Proof of Theorem~\usref{ThmNNo} for $\sigma\notin(\frac{a m}{r_c^2+a^2},\frac{a m}{r_e^2+a^2})$]
  The ODE~\eqref{EqNODE} can be written as $(\mu R')'+V R=0$ where $V(r)=\frac{q(r)^2}{\mu(r)}-L$ is real. Given a solution $R$ satisfying~\eqref{EqNBC}, also $\bar R$ satisfies this ODE but with complex-conjugate boundary conditions. By direct differentiation, the Wronskian
  \[
    W(r) := \mu\Im\bigl(R'(r)\ol{R(r)}\bigr)
  \]
  satisfies $W'(r)=0$. In the notation of~\eqref{EqNBC2}, the boundary conditions for $R$ and $\bar R$ then imply
  \[
    \lim_{r\searrow r_e} W(r) = -q(r_e)|B_e|^2 = \lim_{r\nearrow r_c} W(r) = +q(r_c)|B_c|^2.
  \]
  The numerical range of $-q(r)/(b(r^2+a^2))$ (from~\eqref{EqNODE}, with $b=1+\lambda a^2\geq 1$) on $[r_e,r_c]$ is an interval with endpoints $\frac{a m}{r_c^2+a^2}-\sigma$ and $\frac{a m}{r_e^2+a^2}-\sigma$. If this interval is disjoint from $0$, then $q(r_e)$ and $q(r_c)$ have the same sign, and thus we obtain $B_e=B_c=0$; if $0$ is an endpoint of this interval, then $q(r_e)=0$, $q(r_c)\neq 0$ or $q(r_c)=0$, $q(r_e)\neq 0$, and we can only conclude that $B_c=0$ or $B_e=0$.

  Assuming that $q(r_e)\neq 0$ and thus $B_e=0$, we now prove that $R=0$. (The arguments in the case that $q(r_c)\neq 0$ and thus $B_c=0$ are completely analogous.) The ODE
  \begin{equation}
  \label{EqNNoRegSing}
    \mu(\mu R')' + \bigl(q(r)^2-\mu(r)L\bigr)R = 0
  \end{equation}
  is a regular-singular ODE at $r=r_e$; its indicial roots are $\pm i q(r_e)/\mu'(r_e)$ and thus differ by a nonzero non-integer. Thus, every solution of this ODE is a (unique) linear combination of solutions of class $(r-r_e)^{\pm i q(r_e)/\mu'(r_e)}\bigl(1+(r-r_e)\CI([r_e,r_c))\bigr)$. The outgoing condition~\eqref{EqNBC} forbids the ``$+$'' branch, and $B_e=0$ forbids the ``$-$'' branch. Therefore, $R=0$ on $(r_e,r_c)$.
\end{proof}

Only when $a\neq 0$ (i.e., the black hole has non-vanishing angular momentum) and $m\neq 0$ (i.e., for non-axisymmetric modes) does this argument not cover all $\sigma\in\R$.

%%%%%%%%%%%%%%%%%%%%%%%%%%%%%%%%%%%%%%%%%%%%%%%%%%%%%%%%%%%%%%%%%%%%%%
\section{Proof of Theorem~\ref{ThmNNo} in the superradiant frequency range}
\label{SF}

\textit{We use the notation and assumptions of Theorem~\usref{ThmNNo}}, and recall that $a,m\geq 0$. We now come to the heart of the paper. Suppose $\sigma\in(\frac{a m}{r_c^2+a^2},\frac{a m}{r_e^2+a^2})$ (which can only occur for $m\neq 0$), so
\begin{equation}
\label{EqFsigma}
  \sigma = \frac{a m}{r_0^2+a^2}>0,\quad r_0\in(r_e,r_c);\quad \text{equivalently,}\ q(r_0)=0.
\end{equation}
Since $q(r)=b a m\frac{r^2-r_0^2}{r_0^2+a^2}$, we then have
\[
  q(r_e) < 0 = q(r_0) < q(r_c);\quad q(r) = q_2(r^2-r_0^2),\ q_2:=b\sigma = \frac12 q''>0.
\]
Our strategy is to find a current $J(r)$, i.e., a quadratic function in $R$ and $R'$, that is non-increasing in $r$ and whose boundary values at $r=r_e$, resp.\ $r_c$ pick up a non-positive, resp.\ non-negative multiple of $|B_e|^2$, resp.\ $|B_c|^2$.

%%%%%%%%%%%%%%%%%%%%%%%%%%%%%%%%%%%%%%%%%%%%%%%%%%
\subsection{Derivation of the event-side and cosmological-side currents}
\label{SsFC}

The zero $r_0$ of $q(r)$ is a convenient place to split $[r_e,r_c]=[r_e,r_0]\cup[r_0,r_c]$. From~\eqref{EqNBC}, we see that on $[r_e,r_0]$ (bottom sign) and $[r_0,r_c]$ (top sign),
\begin{equation}
\label{EqFCRpm}
  R_\pm(r) := \exp\biggl(\mp i\int_{r_0}^r \frac{q}{\mu}\biggr)R(r)
\end{equation}
is smooth down to $r_e$ and $r_c$, respectively. Since $\pa_r R=e^{\pm i\int q/\mu}(\pa_r\pm i\frac{q}{\mu})R_\pm$, the ODE~\eqref{EqNODE} for $R$ becomes
\begin{equation}
\label{EqFCConj}
  \bigl(-\pa_r\mu\pa_r + L \mp i(\pa_r q+q\pa_r)\bigr)R_\pm = 0;
\end{equation}
while no longer singular at the respective horizon, this features a (symmetric) first-order term $\mp 2 i\nabla_q$,
\[
  \nabla_q := \frac12(\pa_r q + q\pa_r) = q\pa_r + \frac12 q'.
\]
A natural multiplier for~\eqref{EqFCConj} is $2 f\nabla_q$ where the real-valued function $f=f(r)$ is to be determined. Thus, multiply~\eqref{EqFCConj} by $2 f\overline{\nabla_q R_\pm}$, take real parts (so the contribution from $\mp 2 i\nabla_q$ in~\eqref{EqFCConj} vanishes pointwise), and integrate by parts; this yields:

\begin{lemma}[Currents and bulk terms]
\label{LemmaFC}
  For a function $u=u(r)$ and a real-valued $\cC^2$-function $f=f(r)$ on $(r_e,r_0]$ or $[r_0,r_c)$, define
  \begin{subequations}
  \begin{equation}
  \label{EqFCCur}
  \begin{split}
    J_\pm^f[u] &:= -\mu f q|u'|^2 - \mu f q'\Re(u'\ol{u}) + \Bigl(L f q+\frac12\mu(f q')'\Bigr)|u|^2, \\
    \cF_\pm^f[u] &:= C^f|u'|^2 + (\cD_L f)|u|^2.
  \end{split}
  \end{equation}
  where we set
  \begin{equation}
  \label{EqFCCur2}
    C^f = C^f(r) := \mu' f q - 2\mu f q'-\mu f' q,\quad
    \cD_L f := L q f' + \frac12\bigl(\mu(f q')'\bigr)'.
  \end{equation}
  \end{subequations}
  Then
  \begin{equation}
  \label{EqFCCur3}
    \frac{\dd}{\dd r}J_\pm^f[R_\pm] = \cF_\pm^f[R_\pm].
  \end{equation}
\end{lemma}
\begin{proof}
  Write $u=R_\pm$ for better readability. We expand $0=\Re\int_{r_1}^{r_2} \bigl(-(\mu u')'+L u\mp 2 i\nabla_q u\bigr)\cdot 2 f\nabla_q\ol{R}\,\dd r$ into the sum
  \begin{align*}
    0 &= -\Re \int_{r_1}^{r_2} (\mu R')'\cdot 2 f q \ol{R'} \,\dd r - \Re\int_{r_1}^{r_2} (\mu R')'\cdot f q'\ol{R}\,\dd r + \Re\int_{r_1}^{r_2} R\cdot L(2 f q\ol{R'}+f q'\ol{R})\,\dd r \\
      &\qquad + \Re\int_{r_1}^{r_2} (\mp 2 i\nabla_q R)\cdot 2 f\nabla_q\ol{R}\,\dd r.
  \end{align*}
  The final term vanishes since $f$  is real. Integrating by parts in the first term gives
  \begin{align*}
    &[-2\mu f q|R'|^2]_{r_1}^{r_2} + \int_{r_1}^{r_2} 2(f q)'\mu|R'|^2 + \mu f q\Re(2 R'\ol{R''})\,\dd r \\
    &\quad = [-\mu f q|R'|^2]_{r_1}^{r_2} + \int_{r_1}^{r_2} \bigl(\mu(f q)'-\mu'f q\bigr)|R'|^2\,\dd r
  \end{align*}
  In the passage to the second line, we used $\Re(2 R'\ol{R''})=(|R'|^2)'$ and integrated by parts. Next, we write the second term as
  \begin{align*}
    &[-\mu f q' \Re(R'\ol{R})]_{r_1}^{r_2} + \int_{r_1}^{r_2} \mu f q'|R'|^2 + \mu(f q')'\Re(R'\ol{R})\,\dd r \\
    &\quad = \Bigl[ -\mu f q'\Re(R'\ol{R}) + \frac12\mu(f q')'|R|^2 \Bigr]_{r_1}^{r_2} + \int_{r_1}^{r_2} \Bigl(\mu f q'|R'|^2 - \frac12 \bigl(\mu(f q')'\bigr)'|R|^2\Bigr)\,\dd r
  \end{align*}
  The third term finally is
  \begin{align*}
    \int_{r_1}^{r_2} L f q (|R|^2)' + L f q'|R|^2\,\dd r = [L f q|R|^2]_{r_1}^{r_2} - \int_{r_1}^{r_2} L q f'|R|^2\,\dd r.
  \end{align*}
  Summing the three nonzero terms yields the identity $\int_{r_1}^{r_2} \cF_\pm^f[R]\,\dd r - \bigl(J_\pm^f[R](r_2) - J_\pm^f[R](r_1)\bigr) = 0$. Its derivative in $r_2$ gives the pointwise statement of the lemma.
\end{proof}

We will find $f$ such that $C^f,\cD_L f\leq 0$, so $\cF_\pm^f[u]\leq 0$, with suitable behavior at the juncture $r=r_0$. This will relate the limits of $J_\pm^f[u]$ as $r\searrow r_e$, resp.\ $r\nearrow r_c$; we compute these:

\begin{lemma}[Limiting values of the current]
\label{LemmaFCLim}
  Suppose that $f$ is $\cC^2$ on $(r_e,r_0]$ or $[r_0,r_c)$, continuous down to $r_e$ or $r_c$, and satisfies $|f'(r)| = o\bigl(|r-r_{e/c}|^{-1}\bigr)$, $r\to r_{e/c}$. Then, recalling~\eqref{EqFCRpm} and~\eqref{EqNBC2},
  \begin{equation}
  \label{EqFCLim}
    \lim_{r\searrow r_e} J_-^f[R_-] = L f(r_e)q(r_e)|B_e|^2, \quad \lim_{r\nearrow r_c} J_+^f[R_+] = L f(r_c)q(r_c)|B_c|^2.
  \end{equation}
  Furthermore,
  \begin{equation}
  \label{EqFCLimSeam}
    J_\pm^f[u](r_0) = -2 q_2 r_0\mu(r_0)f(r_0)\Re(R'\ol{R})(r_0) + q_2\mu(r_0)(r f)'(r_0)|R(r_0)|^2
  \end{equation}
\end{lemma}
\begin{proof}
  We will use $q'=2 q_2 r$ and $q(r_0)=0$. In the expression for $J_\pm^f[R_\pm]$ in~\eqref{EqFCCur}, the first two terms vanish at $r=r_{e/c}$. In the third term, $\mu(f q')'=2 q_2\mu(f r)'=2 q_2(\mu f+r\mu f')$ vanishes at $r=r_{e/c}$, so only $L f q|R_\pm|^2$ survives there. For~\eqref{EqFCLimSeam}, we note that $R'(r_0)=R_\pm'(r_0)$ by~\eqref{EqFCRpm}.
\end{proof}

%%%%%%%%%%%%%%%%%%%%%%%%%%%%%%%%%%%%%%%%%%%%%%%%%%
\subsection{Multiplier on the interval \texorpdfstring{$[r_0,r_c]$}{[r0,rc]}}
\label{SsFCc}

On the interval $[r_0,r_c]$, we make the choice
\[
  f_+(r) = \frac{r_0}{r}.
\]
This has the convenient property $(r f_+)'=0$, and hence the second term in~\eqref{EqFCLimSeam} vanishes.

\begin{lemma}[Signs]
\label{LemmaFCc}
  We have $C^{f_+}<0$ and $\cD_L f_+\leq 0$ on $[r_0,r_c]$.
\end{lemma}
\begin{proof}
  Since $(f_+ q')'=2 q_2(r f_+)'=0$, the second statement follows from
  \[
    r^2\cD_L f_+ = -L q r_0 = -L r_0 q_2(r^2-r_0^2) \leq 0,\quad r\geq r_0.
  \]
  For the first statement, we compute
  \[
    -\frac{r}{r_0 q_2}C^{f_+} = -(r^2-r_0^2)\mu' + 4 r\mu-\frac{(r^2-r_0^2)\mu}{r} = (r+r_0)(2\mu-(r-r_0)\mu') + \frac{(r-r_0)^2\mu}{r}.
  \]
  Lemma~\ref{LemmaGIneq1} shows that the first term is positive, and the second term is non-negative for $r\geq r_0$.
\end{proof}

\begin{cor}[Inequality at the seam]
\label{CorFCc}
  We have
  \begin{equation}
  \label{EqFCc}
    \Re(R'\ol{R})(r_0)\leq -\frac{L(r_c^2-r_0^2)}{2 r_c\mu(r_0)}|B_c|^2 \leq 0.
  \end{equation}
\end{cor}
\begin{proof}
  Plug~\eqref{EqFCLim}--\eqref{EqFCLimSeam} (and $(r f_+)'=0$) into $\lim_{r\nearrow r_c}J_+^{f_+}[R_+](r)\leq J_+^{f_+}[R_+](r_0)$.
\end{proof}

%%%%%%%%%%%%%%%%%%%%%%%%%%%%%%%%%%%%%%%%%%%%%%%%%%
\subsection{Multiplier on the interval \texorpdfstring{$[r_e,r_0]$}{[re,r0]}}
\label{SsFCe}

This is the delicate part of the construction.

\begin{prop}[Existence of a multiplier producing a monotone current]
\label{PropFCe}
  For every eigenvalue $L$ of $\slDelta_{\sigma,m}$, there exists a (piecewise smooth) function $f_-\in\cC^2((r_e,r_0])\cap\cC^0([r_e,r_0])$ such that
  \begin{enumerate}
  \item\label{ItFCe1} $C^{f_-}\leq 0$ and $\cD_L f_-\leq 0$;
  \item\label{ItFCe2} $f_-(r_e)=0$ and $|f_-'(r)|=o((r-r_e)^{-1})$ as $r\searrow r_e$;
  \item\label{ItFCe3} $f_-(r)>0$ for $r\in(r_e,r_0]$, and $(r f_-)'(r_0)\geq 0$.
  \end{enumerate}
\end{prop}

\begin{proof}[Proof of Theorem~\usref{ThmNNo} given Proposition~\usref{PropFCe}]
  Lemmas~\ref{LemmaFC} and~\ref{LemmaFCLim} imply that
  \begin{equation}
  \label{EqFCePf}
  \begin{split}
    J_-^{f_-}[R_-](r_0) &= -2 q_2 r_0\mu(r_0)f_-(r_0)\Re(R'\ol{R})(r_0) + q_2\mu(r_0)(r f_-)'(r_0)|R(r_0)|^2 \\
      &\leq \lim_{r\searrow r_e}J_-^{f_-}[R_-](r) = 0;
  \end{split}
  \end{equation}
  for the right-hand side we use the vanishing of $f_-(r_e)$. Since $q_2,\mu(r_0)>0$, this implies that $\Re(R'\ol{R})(r_0)\geq 0$. This is compatible with~\eqref{EqFCc} only if $\Re(R'\ol{R})(r_0)=0$, and this then forces $B_c=0$. The regular-singular ODE argument following~\eqref{EqNNoRegSing} implies that $R=0$. (Alternatively, we can infer that $\int_{r_0}^{r_c}\cF_+^{f_+}[R_+]\,\dd r=0$, and then Lemma~\ref{LemmaFCc} and $R_+(r_c)=0$ imply $R_+=0$ on $[r_0,r_c]$, and thus $R=0$ by uniqueness for \emph{regular} ODEs.)
\end{proof}

\begin{rmk}[Reformulation with a single current]
\label{RmkFCeSingle}
  Define $f,J_\sharp^f$ to be $f_-,J_-^{f_-}$ on $[r_e,r_0]$ and $f_+,J_+^{f_+}$ on $[r_0,r_c]$. Scaling $f_+$, the resulting $f$ is continuous across $r=r_0$. Then
  \[
    \frac{\dd}{\dd r}J_\sharp^f = C^f\Bigl|R'-i\frac{|q|}{\mu}R\Bigr|^2 + (\cD_L f)|R|^2,
  \]
  since $q<0$ on $[r_e,r_0)$ and $q>0$ on $(r_0,r_c]$. Instead of recording a current $\sim\Re(R'\ol{R})=\frac12(|R|^2)'$ at $r=r_0$, we now have a distributional flux across $r=r_0$ arising from the second order term $q_2 r\mu f''$ of $\cD_L f$; since $(r f_+)'|_{r=r_0}=0\leq (r f_-)'|_{r=r_0}$, the jump $J_\sharp^f(r_0+)-J_\sharp^f(r_0-)$ is a \emph{non-positive} multiple of $|R(r_0)|^2$. From $|B_c|^2\sim J_\sharp^f(r_c)\leq J_\sharp^f(r_e)=0$, we then conclude $R=0$ as before.
\end{rmk}

In the remainder of this section, we prove Proposition~\ref{PropFCe}. We write $f_-=f$ for the sought-after function, with $f>0$ on $(r_e,r_0]$. \emph{We work with $r\in[r_e,r_0]$ throughout.} We first observe:

\begin{lemma}[Equivalent formulation of monotonicity]
\label{LemmaFCeEquiv}
  Define $Q(r):=r_0^2-r^2=-q(r)/q_2$ (which is positive on $[r_e,r_0)$) and, for $r\in[r_e,r_0)$,
  \begin{equation}
  \label{EqFCeEquivFA}
    F(r) := \frac{\mu(r)}{Q(r)^2},\quad
    A(r) := \frac{r(4 r\mu + Q\mu')}{Q^2}.
  \end{equation}
  Define $\tau=\tau(r)$ by $\frac{\dd\tau}{\dd r}=\frac{Q}{r\mu}$, and denote differentiation along $\tau$ by a dot, so $\dot f=\frac{r\mu}{Q}f'$. Define the normalized logarithmic derivative $\alpha$ of $f$ by
  \[
    \alpha := \frac{\dot f}{f} = \frac{r\mu}{Q}\frac{f'}{f}.
  \]
  Finally, define the function
  \begin{equation}
  \label{EqFCeEquivSL}
    S_L \colon [r_e,r_0)\times\R \ni (r,\alpha) \mapsto \alpha(L-\alpha) + F(r)\bigl((3 r^2-r_0^2)\alpha-r\mu'(r)\bigr).
  \end{equation}
  Then
  \begin{subequations}
  \begin{alignat}{2}
  \label{EqFCeEquiv1}
    C^f(r)&=\frac{q_2 f Q^2}{r}(\alpha-A) &&\leq 0 \iff \alpha(r)\leq A(r); \\
  \label{EqFCeEquiv2}
    (\cD_L f)(r)&=\frac{q_2 f Q^2}{r\mu}\bigl(\dot\alpha-S_L(r,\alpha)\bigr) && \leq 0 \iff \dot\alpha(r) \leq S_L(r,\alpha(r)).
  \end{alignat}
  \end{subequations}
\end{lemma}
\begin{proof}
  Plugging $q=-q_2 Q$ and $f'=\frac{Q}{r\mu}\dot f=\frac{\alpha Q}{r\mu}f$ into~\eqref{EqFCCur2} gives
  \[
    C^f = -q_2\mu' f Q - 4 q_2 r \mu f + q_2 \frac{\alpha Q^2}{r}f = \frac{q_2 f Q^2}{r}(\alpha-A).
  \]
  Repeatedly using the relationship $X'=\frac{Q}{r\mu}\dot X$, we moreover compute
  \begin{align*}
    -\frac{1}{q_2 f}\cD_L f &= L Q\frac{f'}{f} - \frac{1}{f}\bigl(\mu(r f)'\bigr)' = \frac{L Q^2}{r\mu}\alpha - \frac{1}{f}\bigl((Q\alpha+\mu)f\bigr)' \\
      &= \frac{L Q^2}{r\mu}\alpha + \Bigl(2 r\alpha - \frac{Q^2}{r\mu}\dot\alpha - \mu'\Bigr) - (Q\alpha+\mu)\frac{Q\alpha}{r\mu} \\
      &= \frac{Q^2}{r\mu}\bigl( -\dot\alpha + S_L(r,\alpha) \bigr).
  \end{align*}
  The equivalences in~\eqref{EqFCeEquiv1}--\eqref{EqFCeEquiv2} follow from $q_2>0$ and $f,\mu>0$ on $(r_e,r_0]$.
\end{proof}

This rephrases Proposition~\ref{PropFCe} as a constrained ordinary differential inequality for $\alpha$. One can, of course, recover $f$ from $\alpha$ by integration:

\begin{lemma}[Integration of $\alpha$]
\label{LemmaFCeInt}
  If $\alpha\in\cC^0([r_e,r_0))$ satisfies $\alpha(r_e)>0$, every nonzero solution $f\colon(r_e,r_0)\to\R$ of $\alpha=\dot f/f$ satisfies $\lim_{r\searrow r_e}f(r)=0$ and $|f'(r)|=o((r-r_e)^{-1})$ as $r\searrow r_e$. Moreover, $f$ is unique up to scaling.
\end{lemma}
\begin{proof}
  Direct integration gives $\tau(r)-\frac{r_0^2-r_e^2}{r_e\mu'(r_e)}\log(r-r_e)=\tau(r)-\frac{\log(r-r_e)}{A(r_e)}\in\CI([r_e,r_0])$. (In particular, $\tau(r)\to-\infty$ as $r\searrow r_e$.) Furthermore, integration of $\frac{\dd}{\dd\tau}\log f=\alpha$ gives $f(\tau)=f(\tau_0)\exp(\int_{\tau_0}^\tau\alpha(s)\,\dd s)$, so $|f|\leq C_\eps e^{(\alpha(r_e)-\eps)\tau}\leq C'_\eps(r-r_e)^{(\alpha(r_e)-\eps)/A(r_e)}$ for all $\eps>0$. Therefore, $|f'|=|\frac{Q\alpha}{r\mu}f|=\cO((r-r_e)^{(\alpha(r_e)-\eps)/A(r_e)-1})$. The asymptotics of $f$ follow from $\alpha(r_e)>0$ and
  \begin{equation}
  \label{EqFCeIntApos}
    A(r_e)=\frac{r_e\mu'(r_e)}{r_0^2-r_e^2}>0.
  \end{equation}
  Uniqueness up to scaling follows from the first order nature of the ODE $\dot f/f=\alpha$.
\end{proof}

For later use, we record a monotonicity property of $A$.

\begin{lemma}[Monotonicity of $A$]
\label{LemmaFCeAMono}
  $A(r)=\frac{r(4 r\mu(r)+Q(r)\mu'(r))}{Q(r)^2}$ satisfies $A'(r)\geq 0$ on $[r_e,r_0)$ (with strict inequality on $(r_e,r_0)$), and $A(r_e)>0$ as well as $\lim_{r\nearrow r_0}A(r)=+\infty$.
\end{lemma}
\begin{proof}
  The positivity of $A(r_e)$ was already observed in~\eqref{EqFCeIntApos}. We have
  \[
    A'(r) = \frac{A_1(r)}{Q(r)^3},\quad
    A_1(r) = 8(r_0^2+r^2)r\mu+(r_0^2+5 r^2)Q\mu'+r Q^2\mu'',
  \]
  and thus need to verify the non-negativity of $A_1(r)$. We do this using a Bernstein expansion argument. It is convenient to use~\eqref{EqGNormalize}, so normalize to $r_e=1$, and write
  \begin{equation}
  \label{EqFCeAMonosxyz}
    r_C=1-s < r_e=1 < r=1+x < r_0=1+x+y < r_c=1+x+y+z
  \end{equation}
  with $s\in[0,1]$ and $x,y,z\geq 0$. If we replace $\mu$ by $\lambda^{-1}\mu$, then $A$ scales by $\lambda^{-1}$ as well; thus, it suffices to verify $A'\geq 0$ for
  \begin{equation}
  \label{EqFCeAMonoMu}
    -(r-r_-)(r-r_C)(r-r_e)(r-r_c) = (4+2 x+y+z-s)(x+s)x(y+z)
  \end{equation}
  in place of $\mu$. We also record
  \begin{equation}
  \label{EqFCeAMonoQ}
    Q=y(2+2 x+y).
  \end{equation}
  The derivative ${}'=\pa_r$, with $r_C,r_e,r_c,r_0$ fixed, is given by $\pa_x-\pa_y$ in the coordinates $s,x,y,z$. One then finds that $A_1$, as a polynomial in $s,x,y,z$, has degree $2$ in $s$. In its Bernstein expansion in $s$,
  \begin{equation}
  \label{EqFCeAMonoBern}
    A_1(s,x,y,z) = \sum_{j=0}^2 c_j(x,y,z)b_{j,2}(s),\quad b_{j,n}(s):=\binom{n}{j}s^j(1-s)^{n-j}\geq 0\ \ (0\leq s\leq 1),
  \end{equation}
  the coefficients $c_j$ are polynomials in $x,y,z$; and if we expand $c_j(x,y,z)=\sum_{k,l,n}c_{j,k l n}x^k y^l z^n$, then $c_{j,k l n}\geq 0$ for all $k,l,n$ and for all $j=0,\ldots,2$. (See Appendix~\ref{SMath} for the \texttt{Mathematica} code.)

  The strict inequality $A'(r)>0$ on $(r_e,r_0)$ follows from $A_1>0$ for $s\in(0,1)$, which uses two facts: $b_{j,2}(s)>0$ for $s\in(0,1)$, and not all $c_j$ vanish identically, so one is strictly positive at $x,y,z>0$.
\end{proof}

%%%%%%%%%%%%%%%%%%%%%%%%%%%%%%
\subsubsection{Analysis of the ceiling trajectory}

While the ``ceiling trajectory'' $\alpha=A$ (which saturates the inequality~\eqref{EqFCeEquiv1}) is not regular up to $r=r_0$ (since $A$ blows up inverse quadratically as $r\nearrow r_0$), it is nonetheless useful to determine the extent to which it does work for~\eqref{EqFCeEquiv2}. We begin with:

\begin{lemma}[Ceiling trajectory]
\label{LemmaFCeCeil}
  We have $\dot F/F=A$; that is, for $f=F$ we have equality in~\eqref{EqFCeEquiv1}. Furthermore,
  \begin{equation}
  \label{EqFCeCeil}
    \frac{\cD_L F}{F} = \frac{q_2 Q^2}{r\mu}A \wh{S}_L(r),\quad \wh{S}_L(r):=\frac{\dot A-S_L(r,A)}{A}.
  \end{equation}
\end{lemma}
\begin{proof}
  We have $\frac{\dot F}{F} = \frac{r\mu}{Q}(\log F)' = \frac{r\mu'}{Q} - 2\frac{r\mu\cdot(-2 r)}{Q^2} = \frac{r(4 r\mu+Q\mu')}{Q^2} = A$; and~\eqref{EqFCeCeil} is~\eqref{EqFCeEquiv2}.
\end{proof}

When $\wh{S}_L\leq 0$, we have $\dot A\leq S_L(r,A)$ and thus the ceiling trajectory is acceptable; when $\wh{S}_L>0$ on the other hand, the ceiling trajectory increases too fast, and thus $F$ violates $\cD_L F\leq 0$. Let us thus investigate the sign of $\wh{S}_L(r)$ for $r\in(r_e,r_0)$ and its monotonicity properties.

\begin{lemma}[Limits and monotonicity of $\wh{S}_L$]
\label{LemmaFCeMono}
  On $(r_e,r_0)$, we have $\wh{S}_L'(r)>0$, and
  \begin{equation}
  \label{EqFCeMonoLim}
    \wh{S}_L(r_e) = A(r_e) - L,\quad
    \lim_{r\nearrow r_0}\wh{S}_L(r) = +\infty.
  \end{equation}
\end{lemma}
\begin{proof}
  Since $A\in\CI([r_e,r_0))$ and $A(r_e)>0$ by~\eqref{EqFCeIntApos}, we compute, using $\mu(r_e)=0$ and the formula~\eqref{EqFCeEquivSL}, that $\wh{S}_L(r_e)=-\frac{S_L(r_e,A(r_e))}{A(r_e)}=-L+A(r_e)$. To find the limit as $r\nearrow r_0$, we note that, modulo $(r_0-r)^{-1}\CI([r_e,r_0])$, we have $A(r)\equiv\frac{4 r_0^2\mu(r_0)}{(r_0^2-r^2)^2}\equiv\frac{\mu(r_0)}{(r_0-r)^2}$. Since $\mu(r_0)>0$, this implies $\frac{A'}{A}\equiv\frac{2}{r_0-r}$, and hence the second part of~\eqref{EqFCeMonoLim} follows from
  \[
    \wh{S}_L(r) = \frac{r\mu}{Q}\frac{A'}{A} - \frac{S_L(r,A)}{A(r)} \equiv \frac{\mu(r_0)}{(r_0-r)^2} - \Bigl( -A(r)+\frac{\mu\cdot 2 r_0^2 A}{Q^2 A}\Bigr) \equiv \frac{3\mu(r_0)}{2(r_0-r)^2}.
  \]

  We check the monotonicity of $\wh{S}_L(r)$ using a Bernstein expansion argument similarly to the proof of Lemma~\ref{LemmaFCeAMono}. Write $A=\frac{r A_0}{Q^2}$, so $A_0(r)=4 r\mu(r)+Q(r)\mu'(r)$; using $\frac{F}{A}=\frac{\mu}{r A_0}$, we then first expand
  \begin{align*}
    \wh{S}_L &= \frac{r\mu}{Q}\Bigl(\frac{1}{r}-\frac{2 Q'}{Q}+\frac{A_0'}{A_0}\Bigr) - L + \frac{r A_0}{Q^2} - \frac{\mu}{r A_0}\Bigl((3 r^2-r_0^2)\frac{r A_0}{Q^2} - r\mu'\Bigr) \\
      &= -L + \frac{N(r)}{Q(r)^2 A_0(r)},\quad
      N(r) := 2 r_0^2\mu A_0+r A_0^2+r\mu A_0'Q+\mu\mu' Q^2,
  \end{align*}
  where in the passage to the second line we use $Q=r_0^2-r^2$ and $Q'=-2 r$. Therefore,
  \[
    \wh{S}_L' = \frac{N_1(r)}{Q(r)^3 A_0(r)^2},\quad N_1(r):=N'Q A_0-N Q A_0'-2 N Q'A_0.
  \]
  We need to show that $N_1\geq 0$. Note now that $N_1$ is a polynomial in $r$, the roots $r_C,r_e,r_c$ (and $r_-=-r_C-r_e-r_c$), and the number $r_0\in(r_e,r_c)$. We use~\eqref{EqGNormalize} and define $s,x,y,z$ as in~\eqref{EqFCeAMonosxyz}. If we replace $\mu$ by $\lambda^{-1}\mu$ (also in the definition of $A_0$), then $N_1$ scales by $\lambda^{-3}$; thus, it suffices to verify $N_1\geq 0$ for~\eqref{EqFCeAMonoMu} in place of $\mu$, and with $Q$ given by~\eqref{EqFCeAMonoQ}. Recall that $'=\pa_r$ for fixed $r_C,r_e,r_c,r_0$ reads $\pa_x-\pa_y$ in the coordinates $s,x,y,z$. One can then compute $N_1$ explicitly as a polynomial in $s,x,y,z$. It has degree $6$ in $s$. In its Bernstein expansion in $s$ (cf.\ \eqref{EqFCeAMonoBern}),
  \[
    N_1(s,x,y,z) = \sum_{j=0}^6 c_j(x,y,z)b_{j,6}(s),\quad b_{j,6}(s):=\binom{6}{j}s^j(1-s)^{6-j}\geq 0\ \ (0\leq s\leq 1),
  \]
  the coefficients $c_j$ are polynomials in $x,y,z$; and the coefficient of every monomial $x^k y^l z^n$ appearing in $c_j$ is $\geq 0$. (See Appendix~\ref{SMath} for the \texttt{Mathematica} code that verifies this.) The strict lower bound $N_1>0$ on $(r_e,r_0)$ follows by the same argument as at the end of the proof of Lemma~\ref{LemmaFCeAMono}.
\end{proof}

If $A(r_e)\geq L$, then $\wh{S}_L$ is strictly positive on $(r_e,r_0)$, and thus $f=F$ (i.e., $\alpha=A$) is nowhere admissible for~\eqref{EqFCeEquiv2}; in this case, we show in~\S\ref{SssCeI} that a particular \emph{constant} choice of $\alpha$ verifies~\eqref{EqFCeEquiv1}--\eqref{EqFCeEquiv2}. Otherwise, $f=F$ is admissible for some proper subinterval $(r_e,\wh{r})\subsetneq(r_e,r_0)$ where $\wh{r}$ is the unique zero of $\wh{S}_L$, i.e., $\wh{S}_L(\wh{r})=0$; in this case, we follow $f=F$ until $r=\wh{r}$, and continue $f$ past $r=\wh{r}$ as a solution of $\dot\alpha=S_L(r,\alpha)$ (thus now saturating~\eqref{EqFCeEquiv2}); see~\S\ref{SssCeII}.

%%%%%%%%%%%%%%%%%%%%%%%%%%%%%%
\subsubsection{Case I: low angular eigenvalues (\texorpdfstring{$A(r_e)\geq L$}{A(re) geq L})}
\label{SssCeI}

We first record a lower bound on the angular eigenvalue $L$. Recall that $1\leq m\in\N$.

\begin{lemma}[Lower bound on angular eigenvalues]
\label{LemmaCeILower}
  For $m\geq 1$, all eigenvalues $L$ of $\slDelta_{\sigma,m}$ satisfy
  \begin{equation}
  \label{EqCeILower}
    L \geq \ubar L := \frac{2 r_0^2}{r_0^2 + r_C r_e} > 1.
  \end{equation}
\end{lemma}
\begin{proof}
  Let $\gamma\in\R$. For $u$ in the (dense subset of the) domain of $\slDelta_{\sigma,m}$ such that $(\theta,\phi)\mapsto e^{i m\phi}u$ is smooth and $\|u\|^2_{L^2([0,\pi];\sin\theta\,\|\dd\theta|)}=1$, we can write
  \begin{align*}
    \la\slDelta_{\sigma,m}u,u\ra_{L^2} &= \int_0^\pi c(\theta)|u'|^2\,\sin\theta\,\dd\theta + \int_0^\pi \frac{b^2}{c(\theta)\sin^2\theta}(m-a\sigma\sin^2\theta)^2|u|^2\,\sin\theta\,\dd\theta \\
      &= \int_0^\pi \Bigl|\sqrt{c(\theta)}\,u'-\frac{\gamma\cot\theta}{\sqrt{c(\theta)}}u\Bigr|^2\,\sin\theta\,\dd\theta + 2\gamma\int_0^\pi \Re(u'\ol{u})\cos\theta\,\dd\theta \\
      &\qquad + \int_0^\pi \Bigl(\frac{b^2}{c(\theta)\sin^2\theta}(m-a\sigma\sin^2\theta)^2-\frac{\gamma^2\cot^2\theta}{c(\theta)}\Bigr)|u|^2\,\sin\theta\,\dd\theta.
  \end{align*}
  The first integral is non-negative. The second integral is $\gamma\int_0^\pi (|u|^2)'\cos\theta\,\dd\theta=\gamma[|u|^2\cos\theta]_0^\pi+\gamma\|u\|_{L^2}^2$. Since $m\geq 1$, we have $|u(\theta)|\lesssim|\sin\theta|^m=o(1)$ as $\theta\to 0,\pi$, so the boundary terms vanish. (The fact that the second term is a multiple of the squared $L^2$-norm of $u$ is part of the motivation for completing the square using the $\gamma\cot\theta$ term above.) In the third integral, we insert $\sigma=\frac{a m}{r_0^2+a^2}$ from~\eqref{EqFsigma}; the difference of squares in the integrand then takes on a particularly simple form for $\gamma=b m$: we then obtain
  \begin{align*}
    \la\slDelta_{\sigma,m}u,u\ra_{L^2} &\geq b m + \frac{b^2 m^2}{(r_0^2+a^2)^2}\int_0^\pi \frac{1}{c(\theta)\sin^2\theta}\bigl( (r_0^2+a^2\cos^2\theta)^2 - (r_0^2+a^2)^2\cos^2\theta\bigr)\,\sin\theta\,\dd\theta \\
      &= b m + \frac{b^2 m^2}{(r_0^2+a^2)^2}\int_0^\pi \frac{r_0^4-a^4\cos^2\theta}{1+\lambda a^2\cos^2\theta}|u|^2\,\sin\theta\,\dd\theta.
  \end{align*}
  Since $\lambda>0$ and $a^2<r_C r_e<r_0^2$, the integrand is positive and monotonically decreasing as a function of $\cos^2\theta$, so attains its minimum at $\theta=0,\pi$. Recalling that $1+\lambda a^2=b$, this gives
  \[
    \la\slDelta_{\sigma,m}u,u\ra_{L^2} \geq b m + \frac{b m^2(r_0^2-a^2)}{r_0^2+a^2}.
  \]
  Since $b,m\geq 1$, this is further bounded from below by $1 + \frac{r_0^2-a^2}{r_0^2+a^2} = \frac{2 r_0^2}{r_0^2+a^2}$. This is thus a lower bound for $L$. Corollary~\ref{CorGBound} gives the first bound in~\eqref{EqCeILower}. For the second bound, use $r_0>r_C,r_e$.
\end{proof}

The construction of the function $f$ (or $\alpha=\dot f/f$) for $A(r_e)\geq L$ will, in fact, work under the assumption $A(r_e) > \frac{L+1}{2}$ that is weaker than $A(r_e)>L$ (since $L>1$). We then claim:

\begin{prop}[Case I: constant $\alpha$ works]
\label{PropCeI}
  Suppose that $A(r_e)>\alpha_*:=\frac{L+1}{2}$. Define $f_-$ by $\frac{\dot f_-}{f_-}=\alpha_*$ (i.e., $\frac{f_-'}{f_-}=\frac{\alpha_* Q}{r\mu}$) and $f_-(r_0)=1$. Then $f_-$ satisfies the requirements of Proposition~\usref{PropFCe}.
\end{prop}
\begin{proof}
  Since $\alpha_*>0$ and $Q(r)\geq 0$, with equality at $r=r_0$, part~\eqref{ItFCe2} of Proposition~\ref{PropFCe} follows from Lemma~\ref{LemmaFCeInt}, and part~\eqref{ItFCe3} uses, in addition, the observation $(r f_-)'(r_0)=r_0 f_-'(r_0)+f_-(r_0)=1>0$. The inequality~\eqref{EqFCeEquiv1} follows for $r\in[r_e,r_0]$ from the monotonicity of $A(r)$ (Lemma~\ref{LemmaFCeAMono}), so
  \[
    A(r) \geq A(r_e) > \alpha_*.
  \]

  Finally, we verify~\eqref{EqFCeEquiv2} on $[r_e,r_0]$, i.e.,
  \[
    4 Q^2(r)S_L(r,\alpha_*) = (L^2-1)Q(r)^2 + 2(L+1)\mu(r)(3 r^2-r_0^2) - 4 r\mu(r)\mu'(r) \geq 0.
  \]
  We first use $r\mu'=-4\lambda r^4+2(1-\lambda a^2)r^2-2\bhm r<2 r^2$ to bound the left-hand side from below by
  \begin{equation}
  \label{EqCeIB}
    (L^2-1)Q(r)^2 + 2\mu(r)B_L(r),\quad B_L(r):=(3 L-1)r^2-(L+1)r_0^2.
  \end{equation}
  If $B_L\geq 0$, this is $\geq 0$, as claimed. If $B_L<0$ on the other hand, we use the upper bound $\mu\leq(r-r_e)(r-r_C)$ from Corollary~\ref{CorGBound}; therefore,~\eqref{EqCeIB} is bounded from below by
  \[
    \tilde B_L(r) := (L^2-1)Q(r)^2 + 2(r-r_e)(r-r_C)B_L(r).
  \]
  We must prove that $\tilde B_L\geq 0$. We do this by exploiting monotonicity \emph{in $L$}: the derivative is
  \[
    \frac12\pa_L\tilde B_L = L Q(r)^2 + (r-r_e)(r-r_C)(3 r^2-r_0^2).
  \]
  This is positive when $3 r^2-r_0^2\geq 0$. When $3 r^2-r_0^2\leq 0$, we use $L>1$ and $(r-r_e)(r-r_C)<r^2$ to obtain $\frac12\pa_L\tilde B_L\geq(r_0^2-r^2)^2+r^2(3 r^2-r_0^2)=4 r_0^4\bigl((r/r_0)^4-\frac34(r/r_0)^2+\frac14\bigr)>0$. Thus, $\tilde B_L$ is strictly increasing in $L$ for $L\geq\ubar L=\frac{2 r_0^2}{r_0^2+r_C r_e}>1$ (recalling~\eqref{EqCeILower}). It remains to prove that $\tilde B_{\ubar L}(r)\geq 0$ for $r\in[r_e,r_0]$, which we do using a Bernstein polynomial argument similarly to Lemma~\ref{LemmaFCeAMono}: normalizing to $r_C=1-s<r_e=1\leq r=1+x\leq r_0=1+x+y$, one finds that $(r_0^2+r_C r_e)^2\tilde B_{\ubar L}(r)$ is a cubic polynomial in $s$ that can thus be expressed as $\sum_{j=0}^3 c_j(x,y)b_{j,3}(s)$ where the coefficients $c_j(x,y)$ are themselves polynomials. When $x\geq y$, resp. $y\geq x$, write $x=y+w$, resp.\ $y=x+w$ with $w\geq 0$; then all monomials of $(y,w)\mapsto c_j(y+w,y)$, resp.\ $(x,w)\mapsto c_j(x,x+w)$ have non-negative coefficients.\footnote{This is very sensitive to the precise form of the lower bound $\ubar L$.} See Appendix~\ref{SMath} for the \texttt{Mathematica} code.
\end{proof}

%%%%%%%%%%%%%%%%%%%%%%%%%%%%%%
\subsubsection{Case II: large angular eigenvalues (\texorpdfstring{$A(r_e)<L$}{A(re)<L})}
\label{SssCeII}

When $A(r_e)<L$, we denote by $\wh{r}\in(r_e,r_0)$ the unique zero (cf.\ Lemma~\ref{LemmaFCeMono}) of $\wh{S}_L(r)$ on $(r_e,r_0)$. Recalling $F=\frac{\mu}{Q^2}$ from~\eqref{EqFCeEquivFA}, define
\[
  f_1(r) := \frac{F(r)}{F(\wh{r})},\quad r\in[r_e,\wh{r}] \subsetneq [r_e,r_0].
\]
Then Lemma~\ref{LemmaFCeCeil} and~\eqref{EqFCeCeil} (and the definition of $\wh{r}$) give
\[
  C^{f_1} = 0\ \ \text{on}\ \ [r_e,\wh{r}],\quad \cD_L f_1 < 0\ \ \text{on}\ \ (r_e,\wh{r}).
\]
Not being allowed to use (scalar multiples) of $F$ beyond $r=\wh{r}$ (which is the equality case in~\eqref{EqFCeEquiv1}), we instead continue $f_1$ by having it solve the equality case in~\eqref{EqFCeEquiv2}:

\begin{lemma}[Continuation of $f_1$]
\label{LemmaCeIICont}
  Define $f_2$ on $[\wh{r},r_0]$ to be the solution of $\cD_L f_2=0$, $f_2(\wh{r})=1$, $f_2'(\wh{r})=\frac{F'(\wh{r})}{F(\wh{r})}$. Then the function $f_-$, defined by $f_1$ on $[r_e,\wh{r}]$ and $f_2$ on $[\wh{r},r_0]$, is $\cC^2$ (and piecewise $\CI$).
\end{lemma}
\begin{proof}
  The initial conditions of $f_1$ and $f_2$ at $r=\wh{r}$ match. Furthermore, the coefficient of $f''$ in $\cD_L f$ (see~\eqref{EqFCCur2}) is $q_2\mu r\neq 0$ at $r=\wh{r}$, and since $\cD_L f_1(\wh{r})=0$ (this being equivalent to $\wh{S}_L(\wh{r})=0$) and $\cD_L f_2=0$, also the second derivatives of $f_1$ and $f_2$ match.
\end{proof}

In the remainder of this section, we will prove:
\begin{prop}[Case II: a piecewise construction]
\label{PropCeII}
  Suppose that $A(r_e)<L$. Then the function $f_-$ defined in Lemma~\usref{LemmaCeIICont} satisfies the requirements of Proposition~\usref{PropFCe}.
\end{prop}

Since $\cD_L f_-\leq 0$ on $[r_e,r_0]$ and $C^{f_1}\leq 0$ on $[r_e,\wh{r}]$ by construction, we only need to verify
\begin{enumerate}[leftmargin=0.7in]
\myitem{ItCeII1}{\ref*{PropCeII}.A} the upper bound $C^{f_2}\leq 0$ on $[\wh{r},r_0]$;
\myitem{ItCeII2}{\ref*{PropCeII}.B} the signs in Proposition~\ref{PropFCe}\eqref{ItFCe3}.
\end{enumerate}
We will do this using comparison arguments for the operator
\[
  (q_2 r\mu)^{-1}\cD_L f = f'' + \Bigl(\frac{\mu'}{\mu}+\frac{2}{r}-L\frac{Q}{r\mu}\Bigr)f' + \frac{\mu'}{r\mu} f
\]
from~\eqref{EqFCCur2} (recalling that $q=-q_2 Q$, $Q=r_0^2-r^2$). We can simplify the form of this by recalling $\tau'=\frac{Q}{r\mu}$ and noting that the $f'$-coefficient is the logarithmic $r$-derivative of $\fp:=\mu r^2 e^{-L\tau}>0$, so
\[
  (q_2 r\mu)^{-1}\cD_L f = \fp^{-1}(\fp f')' + \frac{\mu'}{r\mu}f = \fp^{-1}\sD_L f,\quad \sD_L f:=(\fp f')'+r\mu'e^{-L\tau}f.
\]
The advantage of working with $\sD_L$ is that we have the simple Wronskian identity
\begin{equation}
\label{EqCeIIWronski}
  W[f,h]:=\fp(f h'-f'h) \implies W[f,h]' = f\sD_L h-h\sD_L f.
\end{equation}

\begin{proof}[{Proof of~\eqref{ItCeII1}, assuming $f_2>0$ on {$[\wh{r},r_0]$}}]
  We compare $f_2$ with $F$ from~\eqref{EqFCeEquivFA} on $[\wh{r},r_0)$. At $r=\wh{r}$ we have $W[F,f_2]=W[f_1,f_2]=0$ since $f_1$ and $f_2$ agree to second order there. Moreover, since $\cD_L f_2=0$ (so $\sD_L f_2=0$), the identity~\eqref{EqCeIIWronski} gives
  \[
    W[F,f_2]' = -f_2\sD_L F \leq 0
  \]
  since the sign of $\sD_L F$ equals that of $\wh{S}_L$ and thus is $\geq 0$ on $[\wh{r},r_0)$. By definition of $W[F,f_2]$ and using the assumption that $f_2>0$ on $[\wh{r},r_0]$, this implies that $\frac{f_2'}{f_2} \leq \frac{F'}{F}$, i.e., the normalized logarithmic derivative of $f_2$ is $\leq A=\frac{\dot F}{F}$. This proves~\eqref{EqFCeEquiv1} for $r<r_0$, and by continuity also at $r_0$.
\end{proof}

For the proof of~\eqref{ItCeII2}, we use a different comparison function:
\begin{lemma}[A comparison function]
\label{LemmaCeIIComp}
  Set $\alpha_\flat:=\frac{L-1}{2}<\alpha_*:=\frac{L+1}{2}$. Define $h_\flat(r)>0$ on $[\wh{r},r_0]$ by
  \[
    h_\flat(\wh{r})=\frac{\wh{r}}{r_0},\quad \frac{\dot h_\flat}{h_\flat}=\alpha_\flat
  \]
  (or equivalently $\frac{h_\flat'}{h_\flat}=\frac{\alpha_\flat Q}{r\mu}$), and set $h(r):=\frac{r_0}{r}h_\flat(r)$. Then
  \[
    h(\wh{r})=1,\quad
    (r h)'(r_0)=0,\quad
    \sD_L h\leq 0.
  \]
\end{lemma}
\begin{proof}
  The normalization at $r=\wh{r}$ follows from the definition. Next, $h_\flat'(r_0)=0$ since $Q(r_0)=0$, and therefore $(r h)'=r_0 h_\flat'(r)$ vanishes at $r=r_0$ indeed. Finally, since $\sD_L h=\frac{\fp}{q_2 r\mu}\cD_L h$, it suffices to note that, by a direct computation, we have
  \[
    -4 r\mu q_2^{-1}\frac{\cD_L h}{h} = (L^2-1)Q^2 + 2\mu\bigl((3 L-1)r^2-(L+1)r_0^2\bigr);
  \]
  this is precisely the expression~\eqref{EqCeIB} whose positivity we verified there. (Note that the positivity proof did not use $A(r_e)>\alpha_*$, and thus is valid also in the present setting.)
\end{proof}

\begin{proof}[Proof of~\eqref{ItCeII2}]
  We have $h(\wh{r})=f_2(\wh{r})$; and the Wronskian $W[h,f_2]$ satisfies
  \[
    W[h,f_2]' = h\sD_L f_2-f_2\sD_L h = -f_2\sD_L h\geq 0
  \]
  by Lemma~\ref{LemmaCeIIComp} as long as $f_2>0$ (which is certainly true at least for $r\geq\wh{r}$ close to $\wh{r}$). We verify in Lemma~\ref{LemmaCeII2Sign} below that at $r=\wh{r}$, the quantity
  \begin{equation}
  \label{EqCeIICompW}
    \frac{1}{\fp h f_2}W[h,f_2] = \frac{f_2'}{f_2}-\frac{h'}{h} = \frac{Q}{r\mu} \bigl( A+\frac{\mu}{Q} -\alpha_\flat \bigr)
  \end{equation}
  (where we use that $\frac{f'_2}{f_2}=\frac{Q}{r\mu}\frac{\dot F}{F}=\frac{Q}{r\mu}A$ at $r=\wh{r}$) is strictly positive. Granted this, we deduce that $W[h,f_2]>0$ throughout $[\wh{r},r_0]$, so
  \[
    \Bigl(\frac{f_2}{h}\Bigr)' = \frac{W[h,f_2]}{\fp h^2} > 0.
  \]
  This implies that $f_2$ cannot have a zero before $h$ does; since $h$ is strictly positive on $[\wh{r},r_0]$, we get
  \[
    f_2 > 0\ \ \text{on}\ \ [\wh{r},r_0].
  \]
  Finally, at $r=r_0$, we use $r_0 h'(r_0)=-h(r_0)$ to deduce that
  \[
    0 < \frac{r_0 W[h,f_2](r_0)}{\fp} = h(r_0)\bigl(r_0 f'_2(r_0)+f_2(r_0)\bigr) = h(r_0) (r f_2)'|_{r=r_0}.
  \]
  This finishes the proof of Proposition~\ref{PropCeII}.
\end{proof}

To complete the argument, we need to prove the positivity of~\eqref{EqCeIICompW} at $r=\wh{r}$. This uses:

\begin{lemma}[An inequality]
\label{LemmaCeII2Sign}
  Recalling~\eqref{EqFCeCeil} and~\eqref{EqFCeEquivSL}, write $\wh{S}_L=-L+A+B$, where $A(r)$ was defined in~\eqref{EqFCeEquivFA} and
  \[
    B(r) := \frac{r\mu(r)A'(r)}{Q(r)A(r)} - F(r)(3 r^2-r_0^2) + \frac{F(r)}{A(r)}r\mu'(r)
  \]
  where we recall $F=\frac{\mu}{Q^2}$. Then for all $r\in(r_e,r_0)$, we have
  \begin{equation}
  \label{EqCeII2Sign}
    B < A + \frac{2\mu}{Q} + 1.
  \end{equation}
\end{lemma}

At $r=\wh{r}$, we have $\wh{S}_L=0$ and thus $L=A+B<2 A+2\frac{\mu}{Q}+1$, so $A(\wh{r})>\alpha_\flat-\frac{\mu(\wh{r})}{Q(\wh{r})}$, as desired.

\begin{proof}[Proof of Lemma~\usref{LemmaCeII2Sign}]
  As in the proof of Lemma~\ref{LemmaFCeMono}, introduce $A_0(r)=4 r\mu+Q\mu'$, so $A=\frac{r A_0}{Q^2}$. Inserting this into~\eqref{EqCeII2Sign} and multiplying by $A_0 Q^2$ (and recalling $A_0=\frac{A Q^2}{r}>0$), one finds that~\eqref{EqCeII2Sign} is equivalent to
  \[
    A_0 Q^2 + r A_0^2 + \mu A_0 Q - (r_0^2+r^2)A_0\mu - r\mu A_0' Q - Q^2\mu\mu' > 0.
  \]
  The left-hand side is not homogeneous in the quartic coefficient $\lambda$ of $\mu=-\lambda(r+r_C+r_e+r_c)(r-r_C)(r-r_e)(r-r_c)$; rather, it is of the form
  \[
    \lambda E_0 + \lambda^2 E_1 = \lambda(E_0+\lambda E_1),
  \]
  where $E_0$ comes from the term $A_0 Q^2=r^{-1}Q^4 A$; thus $E_0>0$ on $(r_e,r_0)$ by $A(r_e)>0$ and the monotonicity of $A$ (see Lemma~\ref{LemmaFCeAMono}). Furthermore, since $\lambda w=1-\lambda a^2<1$ in the notation of~\eqref{EqGParamExp} (and recalling that we are considering $a\neq 0$) where $w=r_C^2+r_e^2+r_c^2+r_C r_e+r_C r_c+r_e r_c$, we have $\lambda^{-1}(E_0+\lambda E_1)>w E_0+E_1$, and thus it suffices to prove $w E_0+E_1\geq 0$. This can be done using a Bernstein expansion: $w E_0+E_1$ has degree $4$ in $s$ (in the variables~\eqref{EqFCeAMonosxyz}), and the coefficients of the monomials in $x,y,z$ of each Bernstein coefficient are non-negative. See Appendix~\ref{SMath}.
\end{proof}

As shown in~\S\ref{SsFCe}, this completes the proof of Theorem~\ref{ThmNNo}.

%%%%%%%%%%%%%%%%%%%%%%%%%%%%%%%%%%%%%%%%%%%%%%%%%%%%%%%%%%%%%%%%%%%%%%
\section{Zero modes}
\label{S0}

At zero frequency, there do exist mode solutions: for $\sigma=0$ and $m=0$, and for the angular eigenvalue $L=0$ of $\slDelta_{0,0}$ (with constant eigenfunction), every constant solves~\eqref{EqNODE}. (Without separation of variables, this simply amounts to $\Box_g 1=0$.) These are all:

\begin{lemma}[Zero modes]
\label{Lemma0No}
  For $\sigma=0$, the only mode solutions at zero energy are constants.
\end{lemma}
\begin{proof}
  We consider solutions $R(r)$ of~\eqref{EqNODE} with outgoing boundary conditions~\eqref{EqNBC}. From~\eqref{EqNODE} we have $q(r)=-b a m$. If $a m\neq 0$, then the Wronskian argument in the proof of Theorem~\ref{ThmNNo} applies to give $R=0$. Otherwise the ODE~\eqref{EqNODE} reads $-(\mu R')'+L R=0$, and the outgoing condition means $R\in\CI([r_e,r_c])$. Multiplying by $\ol{R}$ and integrating by parts on $[r_e,r_c]$ gives
  \[
    0 = \int_{r_e}^{r_c} \mu|R'|^2\,\dd r + L\int_{r_e}^{r_c} |R|^2\,\dd r.
  \]
  The boundary term vanishes by the smoothness of $R$ and since $\mu(r_e)=\mu(r_c)=0$. Since $\slDelta_{0,m}$ is positive semidefinite, this implies $R=0$ unless $L=0$---so necessarily $m=0$, and with the angular eigenfunctions of $\slDelta_{0,0}=-\frac{1}{\sin\theta}\pa_\theta c(\theta)\sin\theta\,\pa_\theta$ being constants, as follows from an integration by parts. In this case, $R'=0$ shows that $R$ is constant.
\end{proof}

If $u=\sum_{j=0}^d t_*^j u_j$ were a generalized zero mode with $d\geq 1$ and $u_d\neq 0$, then $\Box_g u=0$ would imply $\wh{\Box_g}(0)u_d=0$, so $u_d=1$ after scaling. The equation for $u_{d-1}$ then reads $d[\Box_g,t_*]u_d+\Box_g u_{d-1}=0$, so $\Box(t_*+u_{d-1}/d)=0$. We show in Proposition~\ref{Prop0NoG} below that this does not admit smooth solutions $u_{d-1}$. This contradiction establishes Theorem~\ref{ThmI} for $\sigma=0$.

\begin{prop}[No nontrivial generalized zero modes]
\label{Prop0NoG}
  Recall $t_*$ from~\eqref{EqICoords}. There does not exist $R\in\CI(X)$ such that $\Box_g(t_*+R)=0$.
\end{prop}
\begin{proof}
  The projection of $R$ to $m=0$ would satisfy the same equation. Let thus $R\in\CI(X)$ be axisymmetric. Then the integral of $\varrho^2\wh{\Box_g}(0)R$ over $[r_e,r_c]\times\Sph^2$ against $\sin\theta\,\dd r\,\dd\theta\,\dd\phi$ vanishes, as follows from~\eqref{EqIBox} upon integrating by parts. To complete the proof, we thus need to show that the total integral of $\varrho^2\Box_g t_*\in\CI(X)$ is nonzero. But from~\eqref{EqIBox} we get
  \[
    \varrho^2\Box_g t_* = -\pa_r\mu\pa_r(t_*-t) = \pa_r\mu\pa_r T = \pa_r (r^2+a^2)b J(r)
  \]
  by~\eqref{EqICoords}, the $r$-integral of which is $\bigl((r_c^2+a^2)+(r_e^2+a^2)\bigr)b>0$.
\end{proof}

%%%%%%%%%%%%%%%%%%%%%%%%%%%%%%%%%%%%%%%%%%%%%%%%%%%%%%%%%%%%%%%%%%%%%%
\section{Mode stability in the closed upper half plane; proof of Theorem~\ref{ThmI}}
\label{SPf}

Theorem~\ref{ThmNNo} and Proposition~\ref{Prop0NoG} establish mode stability of $\Box_g$ on subextremal KdS on the real axis. In the SdS case $a=0$, we recall:

\begin{lemma}[Mode stability for SdS]
\label{LemmaPfSdS}
  Theorem~\usref{ThmI} holds in the case $a=0$.
\end{lemma}
\begin{proof}
  We only need to consider $\Im\sigma>0$. Upon writing a mode solution as $e^{-i\sigma t_*}u(x)=e^{-i\sigma t}v$ where $v(x)=e^{i\sigma T}u(x)$ and $\mu v'$ vanish at the horizons. We can thus multiply the equation
  \[
    r^2\wh{\Box_g}(\sigma)v = -\pa_r\mu\pa_r v + \slDelta v - \frac{r^4}{\mu}\sigma^2 v,\quad\slDelta = -\frac{1}{\sin\theta}\pa_\theta\sin\theta\,\pa_\theta-\frac{1}{\sin^2\theta}\pa_\phi^2,
  \]
  by $\ol{v}$ and integrate by parts over $[r_e,r_c]\times\Sph^2$ to obtain
  \begin{equation}
  \label{EqPfSdS}
    \int_0^{2\pi}\int_0^\pi\int_{r_e}^{r_c} \Bigl(\mu |\pa_r v|^2 + |\slnabla v|^2 - \frac{r^4}{\mu}\sigma^2|v|^2\Bigr)\,\sin\theta\,\dd r\,\dd\theta\,\dd\phi = 0.
  \end{equation}
  The boundary term is absent since $\mu v'\ol{v}$ vanishes at the horizons. If $\sigma\notin i(0,\infty)$, then $\Im(\sigma^2)\neq 0$, so taking imaginary parts gives $v=0$. For $\sigma\in i(0,\infty)$, the integrand is a sum of squares, so we can again conclude that $v=0$.
\end{proof}

We now use a continuity argument (in the spirit of \cite[proof of Theorem~1.7 in \S{3.9}]{HintzKdSMS}) to complete the proof of Theorem~\ref{ThmI}.

\begin{proof}[Proof of Theorem~\usref{ThmI}]
  We first recall the main result of~\cite{PetersenVasySubextremal}. Fix $r_*\in(r_e,r_c)$ with $\mu'(r_*)=0$. Write $\tau_*=t_*$ and $\psi_*=\phi_*-\frac{a}{r_*^2+a^2}t_*$. For $u=u(r,\theta,\psi_*)$ we then have $e^{i\varsigma\tau_*}\Box_g(e^{-i\varsigma\tau_*}u)=:P(\varsigma)u$ where $P(\varsigma)$ is a second order differential operator on $X_*=[r_e-\delta,r_c+\delta]\times\Sph^2_{\theta,\psi_*}\cong X$. Note that if $\Box_g(e^{-i\sigma t_*}e^{i m\phi_*}u(r,\theta))=\Box_g\bigl(e^{-i(\sigma-\frac{a m}{r_*^2+a^2})\tau_*}e^{i m\psi_*}u(r,\theta)\bigr)=0$, where $u\neq 0$, so $\sigma$ is a QNM, then $P(\varsigma)$, $\varsigma:=\sigma-\frac{a m}{r_*^2+a^2}$, has nontrivial nullspace. (Note that the frequency shift is purely real.) Then \cite{PetersenVasySubextremal} proved,\footnote{The function $t_*$ used by Petersen--Vasy is required to have timelike differential. This can be arranged by adding to our $t_*$ an appropriate smooth function $\tilde T\in\CI([r_e-\delta,r_c+\delta])$, see \cite[Remark~1.1]{PetersenVasySubextremal} or \cite[\S{3.1}]{HintzKdSMS}. Since this merely conjugates $P(\varsigma)$ by the smooth (on $X_*$) function $e^{i\varsigma\tilde T(r)}$, it does not affect any of the properties of $P(\varsigma)^{-1}$ used below.} for any fixed subextremal KdS metric:
  \begin{enumerate}
  \item The operator $P(\varsigma)^{-1}$ extends from $\Im\varsigma\gg 1$ finite-meromorphically (i.e., every pole has finite order, and the principal parts at each pole are finite rank operators) to $\Im\varsigma>-1$ as an operator $H^{s-1}(X_*)\to H^s(X_*)$ for sufficiently large $s$.
  \item A number $\varsigma\in\C$ is a pole of $P(\varsigma)^{-1}$ if and only if there exists $m\in\Z$ and a nontrivial mode solution $\sim e^{i m\psi_*}$ in the kernel of $P(\varsigma)$, which in turn is equivalent to $\sigma=\varsigma+\frac{a m}{r_*^2+a^2}$ being a QNM in the sense of Definition~\ref{DefIMode}.
  \item\label{ItI3} There exists $C_0>0$ (which can be chosen locally uniformly as one varies the subextremal KdS parameters) such that all QNMs $\varsigma$ with $\Im\varsigma\geq 0$ satisfy $\Im\varsigma\leq C_0$ (from a simple spacetime energy estimate), $|\Re\varsigma|\leq C_0$ (from delicate high-energy estimates).
  \end{enumerate}
  The finite-meromorphic nature of $P(\varsigma)^{-1}$ implies, in particular, that in the decomposition of $\ker P(\varsigma)$ into its $e^{i m\psi_*}$-subspaces, only finitely many subspaces are nontrivial (and each of them has finite dimension). Thus, there is a finite upper bound $\bar m$ on the largest $|m|\in\N_0$ appearing among the decomposed mode solutions corresponding to putative QNMs in $[-C_0,C_0]+i[0,C_0]$. Upon increasing $C_0$ by $\frac{|a|\bar m}{r_*^2+a^2}$, we conclude that all QNMs (for all subextremal KdS parameters in a small neighborhood of some fixed parameters) in the sense of Definition~\ref{DefIMode} lie in $[-C_0,C_0]+i[0,C_0]$; and by increasing $C_0$ further, if necessary, they indeed lie in
  \begin{equation}
  \label{EqPfRect}
    \cR_{C_0} := \bigl( [-C_0,C_0] + i[0,C_0] \bigr) \setminus \{ |\sigma|<C_0^{-1} \}
  \end{equation}
  unless they are $0$ (since $0$ is a QNM and the set of poles of $P(\varsigma)^{-1}$, shifted by multiples of $m$ to pass to $\sigma$, is discrete).

  Suppose now that for the KdS parameters $\Lambda,\bhm,a$ there exists a QNM $\sigma$ with $\Im\sigma>0$. Pick any continuous path $p\colon[0,1]\ni s\mapsto(\Lambda(s),\bhm(s),a(s))$ of subextremal KdS parameters, with $p(1)=(\Lambda,\bhm,a)$ and $a(0)=0$.\footnote{Concretely, this can be done by using the roots as parameters, cf.\ Corollary~\ref{CorGParam}, so setting $r_C(s)=s r_C$ while keeping $r_e$ and $r_c$ fixed.} Let $\cS\subset[0,1]$ be the set of all $s\in[0,1]$ for which KdS with parameters $p(s)$ has a QNM in the upper half plane. By assumption, $1\in\cS$. Moreover, $\cS$ is relatively open since QNMs move continuously under perturbations (see, e.g., \cite[Appendix~A.2]{HintzThesis}). We prove that $\cS$ is closed. Suppose $(s_j)_{j\in\N}$ is a sequence in $\cS$ with $s_j\to s_0\in[0,1]$ as $j\to\infty$. Let $C_0$ denote the constant from item~\eqref{ItI3} for the KdS parameters $p(s_0)$. Let $\sigma_j$ be a QNM for KdS with parameters $p(s_j)$ with $\Im\sigma_j>0$. Then $\sigma_j\in\cR_{C_0}$ when $j$ is sufficiently large. Since $\cR$ is compact, there exists a convergent subsequence $\sigma_{j_k}\xra{k\to\infty}\sigma_0\in\cR$. But $\sigma_0$ must be a QNM of KdS with parameters $p(s_0)$, since otherwise a neighborhood of $\sigma_0$ would be free of QNMs for all nearby KdS spacetimes, contradicting the fact that $\sigma_j$ enters every neighborhood of $\sigma_0$. Since $0$ is an isolated QNM for all nearby KdS spacetimes (the total rank of resonances in a fixed small disc around $0$ being constant under small perturbations and thus equal to $1$), $\sigma_0$ cannot be $0$; and by Theorem~\ref{ThmNNo}, $\sigma_0$ cannot be real, either; so $\Im\sigma_0>0$.

  We conclude that $\cS=[0,1]$. But $0\in\cS$ contradicts Lemma~\ref{LemmaPfSdS}.
\end{proof}

%%%%%%%%%%%%%%%%%%%%%%%%%%%%%%%%%%%%%%%%%%%%%%%%%%%%%%%%%%%%%%%%%%%%%%
\section{The Klein--Gordon case}
\label{SKG}

We now consider, for $\nu>0$, the operator
\[
  \varrho^2(\Box_g+\lambda\nu) = \varrho^2\Box_g + \lambda\nu r^2 + \lambda\nu a^2\cos^2\theta.
\]
Separating variables as in~\eqref{EqNPDE}--\eqref{EqNAngOp}, we need to show that every outgoing solution of
\begin{align*}
  &-(\mu R')' - \frac{q(r)^2}{\mu(r)}R + V(r)R = 0,\quad V(r) := L + \lambda\nu r^2,\quad L\in\spec(\slDelta_{\sigma,m,\nu}), \\
  &\qquad \slDelta_{\sigma,m,\nu} := \slDelta_{\sigma,m} + \lambda\nu a^2\cos^2\theta,
\end{align*}
is trivial for $\sigma\in\R$ (now including for $\sigma=0$); the relevant spheroidal harmonics are now eigenfunctions of $\slDelta_{\sigma,m,\nu}$ (times $e^{i m\phi}$). The outgoing condition for $R$ is the same as in the case $\nu=0$. The proof for non-superradiant frequencies is the same as in~\S\ref{SN}, so we focus on the case $a,m>0$, $\sigma=\frac{a m}{r_0^2+a^2}>0$, and $q(r_e)<0=q(r_0)<q(r_c)$, $q(r)=q_2(r^2-r_0^2)$, as in~\S\ref{SF}. (The case $\sigma=0$ is discussed in~\S\ref{SsKGPf}.) We define $R_\pm$ by~\eqref{EqFCRpm}; instead of~\eqref{EqFCConj} they satisfy
\[
  (-\pa_r\mu\pa_r + V \mp 2 i\nabla_q)R_\pm = 0.
\]
We use a multiplier $2 f\nabla_q$ for real-valued $f$; the replacement for Lemma~\ref{LemmaFC} is the identity
\begin{align*}
  &J_\pm^f[u] := -\mu f q|u'|^2 - \mu f q'\Re(u'\ol{u}) + \Bigl(V f q + \frac12\mu(f q')'\Bigr)|u|^2 \\
  &\qquad \implies \frac{\dd}{\dd r}J_\pm^f[u] = \cF_\pm^f[u] = C^f|u'|^2 + (\cD_V f)|u|^2, \\
  &\qquad \qquad\qquad C^f=\mu' f q-2\mu f q'-\mu f' q,\quad \cD_V f=q(f V)' + \frac12\bigl(\mu(f q')'\bigr)'.
\end{align*}
The expression for $J_\pm^f[u](r_0)$ is the same as in~\eqref{EqFCLimSeam}.

%%%%%%%%%%%%%%%%%%%%%%%%%%%%%%%%%%%%%%%%%%%%%%%%%%
\subsection{Multiplier and Hardy inequality on the interval \texorpdfstring{$[r_0,r_c]$}{[r0,rc]}}

On $[r_0,r_c]$, we again take $f_+(r)=\frac{r_0}{r}$; then writing $q(r)=q_2\tilde Q(r)$, $\tilde Q(r):=r^2-r_0^2\geq 0$, we compute (cf.\ Lemma~\ref{LemmaFCc})
\begin{alignat*}{2}
  \cD_V f_+ &= q_2 d_+,&\quad d_+&:=-r_0\tilde Q(r)\Bigl(\frac{L}{r^2}-\lambda\nu\Bigr), \\
  C^{f_+} &= -q_2 c_+,&\quad c_+&:=\frac{r_0}{r}\Bigl( -\tilde Q\mu' + 4 r\mu - \frac{\tilde Q\mu}{r} \Bigr) > 0.
\end{alignat*}

\begin{rmk}[Wrong sign]
\label{RmkKGWrong}
  Consider, for the sake of concreteness, the conformally coupled case $\nu=2$. If $L<2\lambda r_c^2$, then $d_+>0$ near $r=r_c$. Numerically, for $r_C=1-\eps$, $r_e=1$, $r_c=\eps^{-1}$, and $r_0=1+\eps$, we have $2\lambda r_c^2\to 2$ but $L\approx 1.2$ for $\eps\ll 1$ and the lowest eigenvalue of $\slDelta_{\sigma,1,2}$, so $\cD_V f_+$ may indeed have the wrong (i.e., positive) sign.
\end{rmk}

We can borrow the definite sign of the term $c_+|R_+'|^2$ of the flux $-q_2^{-1}\cF_+^{f_+}[R_+]$ by writing, for a function $h=h(r)$ that we choose below,
\begin{align*}
  -q_2^{-1}\cF_+^{f_+}[R_+] &= c_+|R_+'|^2 - d_+|R_+|^2 \\
    &= c_+\Bigl| R_+' + \frac{h}{c_+}R_+ \Bigr|^2 - 2 h\Re(R_+'\ol{R_+}) - \Bigl(d_+ + \frac{h^2}{c_+}\Bigr)|R_+|^2 \\
    &= c_+\Bigl| R_+' + \frac{h}{c_+}R_+ \Bigr|^2 + \Bigl( h' - \frac{h^2}{c_+} - d_+ \Bigr)|R_+|^2 - \bigl( h |R_+|^2 \bigr)'.
\end{align*}
The coefficient of $|R_+|^2$ is $h' - \frac{h^2}{c_+} + \frac{L r_0\tilde Q}{r^2} - \lambda\nu r_0\tilde Q$. If we choose $h'(r) = \lambda\nu r_0\tilde Q(r)$, $h(r_0)=0$, so
\[
  h(r)=\lambda\nu r_0\tilde Q_1(r),\quad
  \tilde Q_1(r):=\int_{r_0}^r \tilde Q(s)\,\dd s = \frac13(r^3-r_0^3) - (r-r_0)r_0^2,
\]
the $h'$- and $-\lambda\nu r_0\tilde Q$-terms cancel, so the $|R_+|^2$-coefficient is
\[
  \frac{r_0\tilde Q(r)}{r^2}\Bigl( L - \frac{\nu^2}{D(r)} \Bigr) \geq \frac{r_0\tilde Q(r)}{r^2}\Bigl(1-\frac{\nu^2}{D(r)}\Bigr),\quad D(r):=\frac{c_+(r)\tilde Q(r)}{\lambda^2 r_0 r^2\tilde Q_1(r)^2};
\]
here we use that $L\geq 1$,\footnote{A sharper lower bound for $L$ would ultimately give a larger range of $\nu$; the lower bound $L\geq 1$ suffices to get a range of $\nu$ that includes the conformal mass case, and hence we do not try to optimize this part (also since the present choice of $f_+$ is unlikely to be optimal in any case).} as follows from $\slDelta_{\sigma,m,\nu}\geq\slDelta_{\sigma,m}=\slDelta_{\sigma,m,0}$ (on the level of quadratic forms) and Lemma~\ref{LemmaCeILower}. This, in turn, is non-negative for $\nu^2\leq 18$, i.e., $\nu\leq 3\sqrt{2}$, in view of:

\begin{lemma}[Lower bound]
\label{LemmaWGLowerD}
  We have $D(r)\geq 18$ for $r\in(r_0,r_c]$, and $D(r)\to +\infty$ as $r\searrow r_0$.
\end{lemma}

For the parameters of Remark~\ref{RmkKGWrong}, we have $D(r_c)\to 18$ as $\eps\to 0$.

\begin{proof}[Proof of Lemma~\usref{LemmaWGLowerD}]
  The claimed bound is equivalent to
  \[
    (-r\tilde Q\mu'+4 r^2\mu-\tilde Q\mu)\tilde Q-18\lambda^2 r^4\tilde Q_1(r)^2 \geq 0.
  \]
  The first summand scales with the prefactor $\lambda$ of $\mu$; thus, setting $\mu_0=-(r-r_-)(r-r_C)(r-r_e)(r-r_c)$, we must show
  \begin{equation}
  \label{EqWGLowerD1}
    \lambda D_0 - \lambda^2 D_1 \geq 0,\qquad D_0 := (-r\tilde Q\mu_0'+4 r^2\mu_0-\tilde Q\mu_0)\tilde Q,\quad D_1 := 18 r^4\tilde Q_1^2.
  \end{equation}
  Recall the functions $w,t$ of the roots of $\mu$ from~\eqref{EqGParamExp}, and $1=\lambda w+\lambda^2 t$; here $\lambda t=a^2>0$ since $a\neq 0$. Multiplying~\eqref{EqWGLowerD1} by $t$, we must show
  \[
    \lambda\tilde D_0 - D_1 \geq 0,\quad \tilde D_0 := t D_0+w D_1.
  \]
  Normalize $r_e=1$ and introduce $s\in(0,1]$ and $x,y,z\geq 0$ by
  \[
    r_C = 1-s < r_e = 1 < r_0 = 1+x \leq r = 1+x+y \leq r_c = 1+x+y+z.
  \]
  This differs from~\eqref{EqFCeAMonosxyz} in the order of $r$ and $r_0$, so now $\mu_0 = (4+2 x+2 y+z-s)z(x+y)(x+y+s)$ and $\pa_r=\pa_y-\pa_z$, but still $\tilde Q=(1+x+y)^2-(1+x)^2=y(2+2 x+y)$. One can then check $\tilde D_0\geq 0$ using a Bernstein expansion (see again Appendix~\ref{SMath}). Since $1=\lambda(w+a^2)$, we have $\lambda^{-1}=w+a^2\leq w+r_C=w+1-s$ by Corollary~\ref{CorGBound}, and hence $\tilde D_0\geq 0$ implies
  \[
    \lambda\tilde D_0 - D_1 \geq \frac{\tilde D_0}{w+1-s} - D_1 = \frac{D_2}{w+1-s},\quad D_2 := \tilde D_0-(w+1-s)D_1.
  \]
  A Bernstein expansion (see Appendix~\ref{SMath}) proves $D_2\geq 0$, finishing the proof.
\end{proof}

\emph{We henceforth assume that $\nu\in(0,3\sqrt{2}]$.} Integrating $\frac{\dd}{\dd r}J_+^{f_+}[R_+]=\cF_+^{f_+}[R_+]$ over $[r_0,r_c]$, gives
\begin{align*}
  J_+^{f_+}[R_+](r_c) = V(r_c)f_+(r_c)q(r_c)|B_c|^2 &\leq J_+^{f_+}[R_+](r_0) + q_2\bigl(h|R_+|^2\bigr)\big|_{r_0}^{r_c} \\
    &\leq -2 q_2\mu(r_0) r_0\Re(R_+'\ol{R_+})(r_0) + q_2 h(r_c)|B_c|^2
\end{align*}
where we use~\eqref{EqNBC2}. This yields
\begin{equation}
\label{EqWGCosmRobin}
  \Re(R'\ol{R})(r_0)\leq -C|B_c|^2,\quad C>0,
\end{equation}
analogously to~\eqref{EqFCc} since $q_2,\mu(r_0),h(r_c)>0$ and (using $V(r)\geq\lambda\nu r^2$)
\[
  V(r_c)f_+(r_c)\tilde Q(r_c) - h(r_c) \geq \lambda\nu r_0\bigl( r_c(r_c^2-r_0^2) - \tilde Q_1(r_c)\bigr) = \frac23\lambda\nu r_0(r_c^3-r_0^3) > 0.
\]

%%%%%%%%%%%%%%%%%%%%%%%%%%%%%%%%%%%%%%%%%%%%%%%%%%
\subsection{Multiplier on the interval \texorpdfstring{$[r_e,r_0]$}{[re,r0]}}

We can use the \emph{same} multiplier $2 f_-\nabla_q$ as in the massless case. Indeed, the expression for $C^{f_-}$ is the same as in the massless case (and thus $C^{f_-}\leq 0$), and
\[
  \cD_V f_- = \cD_L f_- - \lambda\nu q_2 Q(r)(r^2 f_-)';
\]
since $\cD_L f_-\leq 0$, it suffices to prove that $(r^2 f_-)'\geq 0$ on $(r_e,r_0]$. In the setting of Proposition~\ref{PropCeI}, this follows from
\[
  \frac{(r^2 f_-)'}{r^2 f_-} = \frac{2}{r} + \frac{\alpha_*Q}{r\mu} > 0.
\]
In the setting of Proposition~\ref{PropCeII}, we use that $f_-$ is a positive constant multiple of $F$ on $[r_e,\wh{r}]$, and $\frac{F'}{F}=\frac{Q}{r\mu}A>0$ by Lemma~\ref{LemmaFCeAMono}. On $[\wh{r},r_0]$ on the other hand, we recall that~\eqref{EqCeIICompW} is positive and $\frac{h'}{h}=-\frac{1}{r}+\frac{\alpha_\flat Q}{r\mu}$ for the function $h$ introduced there, so $\frac{f_2'}{f_2}>-\frac{1}{r}+\frac{\alpha_\flat Q}{r\mu}$, and hence
\[
  \frac{(r^2 f_2)'}{r^2 f_2} > \frac{1}{r} + \frac{\alpha_\flat Q}{r\mu} > 0,
\]
as required. (We also recall that the lower bound~\eqref{EqCeILower} on angular eigenvalues is valid also for $\slDelta_{\sigma,m,\nu}\geq\slDelta_{\sigma,m}$.) Therefore, we obtain~\eqref{EqFCePf} as in the massless case, and hence
\begin{equation}
\label{EqWGEvRobin}
  \Re(R'\ol{R})(r_0) \geq 0.
\end{equation}

%%%%%%%%%%%%%%%%%%%%%%%%%%%%%%%%%%%%%%%%%%%%%%%%%%
\subsection{Conclusion of the proof}
\label{SsKGPf}

The inequalities~\eqref{EqWGCosmRobin} and~\eqref{EqWGEvRobin} force $\Re(R'\ol{R})(r_0)=0$, thus $B_c=0$, and hence $R=0$. \emph{This establishes mode stability for $\sigma\in\R\setminus\{0\}$.} For $\sigma=0$, the same arguments as in the proof of Lemma~\ref{Lemma0No} reduce to $a m=0$, so $q=0$, and then lead to the identity $0=\int_{r_e}^{r_c}\mu|R'|^2\,\dd r+\int_{r_e}^{r_c} V|R|^2\,\dd r$. Since $V\geq\lambda\nu r^2$ is strictly positive for $\nu>0$, this gives $R=0$.

Analogously to Lemma~\ref{LemmaPfSdS}, we also record that Theorem~\ref{ThmII} holds in the case $a=0$ for all $\nu>0$; for $\Im\sigma>0$, the integrand in~\eqref{EqPfSdS} merely has an additional non-negative term $r^2\lambda\nu|v|^2$.

The continuity argument in~\S\ref{SPf}, but now for the family $\Box_g+\lambda\nu$ and without excising a neighborhood of $\sigma=0$ in~\eqref{EqPfRect}, finishes the proof of Theorem~\ref{ThmII}.

%%%%%%%%%%%%%%%%%%%%%%%%%%%%%%%%%%%%%%%%%%%%%%%%%%%%%%%%%%%%%%%%%%%%%%
%%%%%%%%%%%%%%%%%%%%%%%%%%%%%%%%%%%%%%%%%%%%%%%%%%%%%%%%%%%%%%%%%%%%%%
\appendix
\section{Mathematica code}
\label{SMath}

\begin{lstlisting}[style=mathematica,label={lst:positivity}]
(* For a polynomial poly in the variable xx, compute its Bernstein coefficients in the basis of Bernstein polynomials of degree deg. *)
BernsteinCoeff[poly_, xx_, deg_] := Block[{coeff},
  coeff = CoefficientList[poly, xx];
  Table[
    Sum[
      Binomial[j, k]/Binomial[deg, k] coeff[[k + 1]],
      {k, 0, j}
    ],
    {j, 0, deg}
  ]
]

(* mu is the normalized function -(r-r_-)(r-r_C)(r-r_e)(r-r_c). mup is its r-derivative.  Similarly for the other functions. N1fn is the function denoted N_1(r) in the paper. *)
mu[s_, x_, y_, z_] := (4 + 2 x + y + z - s) (x + s) x (y + z);
mup[s_, x_, y_, z_] := D[mu[s, x, y, z], x] - D[mu[s, x, y, z], y];
mupp[s_, x_, y_, z_] := D[mup[s, x, y, z], x] - D[mup[s, x, y, z], y];

Q[x_, y_] := y (2 + 2 x + y);
Qp[x_, y_] := D[Q[x, y], x] - D[Q[x, y], y];

A0[s_, x_, y_, z_] := 4 (1 + x) mu[s, x, y, z] + Q[x, y] mup[s, x, y, z];
A0p[s_, x_, y_, z_] := D[A0[s, x, y, z], x] - D[A0[s, x, y, z], y];

A1[s_, x_, y_, z_] :=
  8 ((1 + x + y)^2 + (1 + x)^2) (1 + x) mu[s, x, y, z]
  + ((1 + x + y)^2 + 5 (1 + x)^2) Q[x, y] mup[s, x, y, z]
  + (1 + x) Q[x, y]^2 mupp[s, x, y, z];

Nfn[s_, x_, y_, z_] :=
  2 (1 + x + y)^2 mu[s, x, y, z] A0[s, x, y, z]
  + (1 + x) A0[s, x, y, z]^2
  + (1 + x) mu[s, x, y, z] * A0p[s, x, y, z] Q[x, y]
  + mu[s, x, y, z] mup[s, x, y, z] Q[x, y]^2;
Nfnp[s_, x_, y_, z_] :=
  D[Nfn[s, x, y, z], x] - D[Nfn[s, x, y, z], y];

N1fn[s_, x_, y_, z_] :=
  Nfnp[s, x, y, z] Q[x, y] A0[s, x, y, z]
  - Nfn[s, x, y, z] Q[x, y] A0p[s, x, y, z]
  - 2 Nfn[s, x, y, z] Qp[x, y] A0[s, x, y, z];

ubarLNum[s_, x_, y_] := 2 (1 + x + y)^2;
ubarLDenom[s_, x_, y_] := (1 + x + y)^2 + (1 - s);

tildeBNum[s_, x_, y_] :=
  (ubarLNum[s, x, y]^2 - ubarLDenom[s, x, y]^2) Q[x, y]^2
  + 2 ubarLDenom[s, x, y] x (x + s)
      ((3 ubarLNum[s, x, y] - ubarLDenom[s, x, y]) (1 + x)^2
       - (ubarLNum[s, x, y] + ubarLDenom[s, x, y]) (1 + x + y)^2);

E0[s_, x_, y_, z_] := A0[s, x, y, z] Q[x, y]^2;
E1[s_, x_, y_, z_] :=
  (1 + x) A0[s, x, y, z]^2
  + mu[s, x, y, z] A0[s, x, y, z] Q[x, y]
  - ((1 + x + y)^2 + (1 + x)^2) A0[s, x, y, z] mu[s, x, y, z]
  - (1 + x) mu[s, x, y, z] A0p[s, x, y, z] Q[x, y]
  - Q[x, y]^2 mu[s, x, y, z] mup[s, x, y, z];
w[s_, x_, y_, z_] :=
  (1 - s)^2 + 1 + (1 + x + y + z)^2 + (1 - s)
  + (1 - s) (1 + x + y + z) + (1 + x + y + z);
Bpf[s_, x_, y_, z_] :=
  w[s, x, y, z] E0[s, x, y, z] + E1[s, x, y, z];

muCosm[s_, x_, y_, z_] := (4 + 2 x + 2 y + z - s) z (x + y) (x + y + s);
mupCosm[s_, x_, y_, z_] := D[muCosm[s, x, y, z], y] - D[muCosm[s, x, y, z], z];

tRoot[s_, x_, y_, z_] := (1 - s) (1 + x + y + z) ((1 - s) + 1 + (1 + x + y + z));
wRoot[s_, x_, y_, z_] :=
  (1 - s)^2 + 1 + (1 + x + y + z)^2 + (1 - s)
  + (1 - s) (1 + x + y + z) + (1 + x + y + z);
Q1[x_, y_] := 1/3 ((1 + x + y)^3 - (1 + x)^3) - y (1 + x)^2;
D0[s_, x_, y_, z_] :=
  (-(1 + x + y) Q[x, y] mupCosm[s, x, y, z]
   + 4 (1 + x + y)^2 muCosm[s, x, y, z]
   - Q[x, y] muCosm[s, x, y, z]) Q[x, y];
D1[s_, x_, y_, z_] := 18 (1 + x + y)^4 Q1[x, y]^2;
tildeD0[s_, x_, y_, z_] := tRoot[s, x, y, z] D0[s, x, y, z] + wRoot[s, x, y, z] D1[s, x, y, z];
D2[s_, x_, y_, z_] := tildeD0[s, x, y, z] - (wRoot[s, x, y, z] + 1 - s) D1[s, x, y, z];

(* Compute the degrees (in s) of A1 and N1: *)
Print[Length[CoefficientList[A1[s, x, y, z], s]] - 1];      (* output: 2 *)
Print[Length[CoefficientList[N1fn[s, x, y, z], s]] - 1];    (* output: 6 *)
Print[Length[CoefficientList[tildeBNum[s, x, y], s]] - 1];  (* output: 3 *)
Print[Length[CoefficientList[Bpf[s, x, y, z], s]] - 1];     (* output: 4 *)
Print[Length[CoefficientList[tildeD0[s, x, y, z], s]] - 1]; (* output: 4 *)
Print[Length[CoefficientList[D2[s, x, y, z], s]] - 1];      (* output: 4 *)

(* Expand A1, N1fn, tildeBNum, Bpf, tildeD0, D2 into the Bernstein basis of degree 2, 6, 3, 4, 4, 4, respectively. Check that all Bernstein coefficients have themselves non-negative coefficients in their monomials in x, y, and (if applicable) z. *)
bA1 = Expand[BernsteinCoeff[A1[s, x, y, z], s, 2]];
bN1fn = Expand[BernsteinCoeff[N1fn[s, x, y, z], s, 6]];
btildeBNum = Expand[BernsteinCoeff[tildeBNum[s, x, y], s, 3]];
bBpf = Expand[BernsteinCoeff[Bpf[s, x, y, z], s, 4]];
btildeD0 = Expand[BernsteinCoeff[tildeD0[s, x, y, z], s, 4]];
bD2 = Expand[BernsteinCoeff[D2[s, x, y, z], s, 4]];

Print[ Table[
    Min[Flatten[CoefficientList[bA1[[j]], {x, y, z}]]], {j, 1, 3}
  ] ];
(* Output: {0, 0, 0} *)

Print[ Table[
    Min[Flatten[CoefficientList[bN1fn[[j]], {x, y, z}]]], {j, 1, 7}
  ] ];
(* Output: {0, 0, 0, 0, 0, 0, 0} *)

Print[ Table[
    Min[Flatten[CoefficientList[bBpf[[j]], {x, y, z}]]], {j, 1, 5}
  ] ];
(* Output: {0, 0, 0, 0, 0} *)

(* For tildeBNum, we separate cases: for x\geq y, write x=y+w,
   and for y\geq x, write y=x+w. *)

Print[ Table[
    Min[Flatten[
      CoefficientList[Expand[btildeBNum[[j]] /. x -> y + w], {y, w}]
    ]], {j, 1, 4}
  ] ];
(* Output: {0, 0, 0, 0} *)

Print[ Table[
    Min[Flatten[
      CoefficientList[Expand[btildeBNum[[j]] /. y -> x + w], {x, w}]
    ]], {j, 1, 4}
  ] ];
(* Output: {0, 0, 0, 0} *)

Print[ Table[
    Min[Flatten[CoefficientList[btildeD0[[j]], {x, y, z}]]], {j, 1, 5}
  ] ];
(* Output: {0, 0, 0, 0, 0} *)

Print[ bD2[[5]] ]; (* Output: 0 *)
Print[ Table[
    Min[Flatten[CoefficientList[bD2[[j]], {x, y, z}]]], {j, 1, 4}
  ] ];
(* Output: {0, 0, 0, 0} *)

\end{lstlisting}

%%%%%%%%%%%%%%%%%%%%%%%%%%%%%%%%%%%%%%%%%%%%%%%%%%%%%%%%%%%%%%%%%%%%%%
\bibliographystyle{alphaurl}

%\bibliography{
%../bib/math.bib,
%../bib/phys.bib,
%../bib/mathcheck.bib
%}

\end{document}